\documentclass[10pt,french,english]{amsart}
\usepackage[T1]{fontenc}
\usepackage{amsthm}
\usepackage{amstext}
\usepackage{amssymb}
\usepackage{esint}
\usepackage[all]{xy}

\makeatletter
\numberwithin{equation}{section} 
\numberwithin{figure}{section} 
\theoremstyle{plain}
\theoremstyle{plain}
\newtheorem{thm}{Theorem}
  \theoremstyle{remark}
  \newtheorem*{rem*}{Remark}
  \theoremstyle{plain}
  \newtheorem{prop}[thm]{Proposition}
  \theoremstyle{plain}
  \newtheorem{cor}[thm]{Corollary}
  \theoremstyle{plain}
  \newtheorem{lem}[thm]{Lemma}

\usepackage{graphicx}
\usepackage{amssymb}
\usepackage{latexsym}
\usepackage[T1]{fontenc}
\usepackage[all]{xy}
\usepackage{stmaryrd}

\newcommand{\moy}[1]{\langle #1 \rangle}
\newcommand{\lp}{\left(}
\newcommand{\rp}{\right)}
\newcommand{\lva}{\left|}
\newcommand{\rva}{\right|}
\newcommand{\lno}{\left\|}
\newcommand{\rno}{\right\|}
\newcommand{\Dg}{\Delta_g}
\newcommand{\fr}{\frac}

\newcommand{\h}{\hbar}
\newcommand{\ti}[1]{\tilde {#1}}
\newcommand{\eps}{\varepsilon}
\newcommand{\trm}[1]{\textrm{#1}}
\def\hto0{\xrightarrow{h\to 0}}

\renewcommand{\d}{\partial}

\newcommand{\norm}[1]{\|#1\|}

\newcommand{\defeq}{\stackrel{\rm{def}}{=}}
\newcommand{\llb}{\llbracket}
\newcommand{\rrb}{\rrbracket}

\newcommand{\bbbone}{{\mathchoice {1\mskip-4mu {\rm{l}}} {1\mskip-4mu {\rm{l}}}{ 1\mskip-4.5mu {\rm{l}}} { 1\mskip-5mu {\rm{l}}}}}

\newcommand{\cA}{\mathcal A}
\newcommand{\cD}{\mathcal D}
\newcommand{\cL}{{\mathcal L}}
\newcommand{\cM}{{\mathcal M}}
\newcommand{\cO}{{\mathcal O}}
\newcommand{\cU}{{\mathcal U}}

\newcommand{\cB}{{\mathcal B}}
\newcommand{\cS}{\mathcal S}

\newcommand{\cH}{{\mathcal H}}
\newcommand{\cP}{{\mathcal P}}
\newcommand{\cE}{{\mathcal E}}
\newcommand{\cV}{{\mathcal V}}
\newcommand{\cW}{{\mathcal W}}
\newcommand{\cK}{{\mathcal K}}
\newcommand{\cT}{\mathcal T}

\newcommand{\cF}{\mathcal F}
\newcommand{\cX}{\mathcal X}

\newcommand{\ka}{\kappa}
\newcommand{\ga}{\gamma}
\newcommand{\be }{\beta }

\newcommand{\bs}{\boldsymbol}

\newcommand{\fkP}{\mathsf P}
\newcommand{\fkV}{\mathsf V}

\newcommand{\msV}{\mathsf V}

\newcommand{\C}{{\mathbb C}}
\newcommand{\N}{{\mathbb N}}

\newcommand{\R}{{\mathbb R}}

\newcommand{\supp}{\operatorname {supp}}

\newcommand{\Spec}{\operatorname{Spec}}

\newcommand{\e}{\operatorname{e}}
\newcommand{\Op}{\operatorname{Op}}
\newcommand{\oph}{\operatorname{Op}_{\h}}
\newcommand{\Id}{\operatorname{Id}}

\newcommand{\Card}{\operatorname{Card}}
\renewcommand{\div}{\operatorname{div}}
\newcommand{\WF}{\operatorname{WF}}
\newcommand{\diam}{\operatorname{diam}}
\renewcommand{\Im }{\operatorname{Im}}
\renewcommand{\Re}{\operatorname{Re}}
\newcommand{\Res}{\operatorname{Res}}
\newcommand{\et}{\operatorname{and}}
\newcommand{\press}{\operatorname{Pr}}
\newcommand{\Jac}{\operatorname{Jac}}
\newcommand{\dvol}{d\textrm{vol}}
\renewcommand{\i}{\operatorname{i}}

\newcommand{\lh}{\log \hbar^{-1}}

\address{
Institut de Physique Th\'eorique\\
CEA Saclay\\
91191 Gif-sur-Yvette\\
France}

\makeatother

\usepackage{babel}
\addto\extrasfrench{\providecommand{\fg}{\ifdim\lastskip>\z@\unskip\fi~\frqq}}

\begin{document}
\global\long\def\trm#1{\textrm{#1}}
\global\long\def\scal#1#2{\langle#1,#2\rangle}
\global\long\def\moy#1{\langle#1\rangle}
\global\long\def\norm#1{\|#1\|}
\global\long\def\PP#1{\mathsf{P}_{\gamma_{#1}}}
\global\long\def\VV#1{\mathsf{V}_{\gamma_{#1}}}
\global\long\def\FF{\mathsf{F}}

\global\long\def\prq{\Pr(q^{u})^{+}}
\global\long\def\pra{\Pr(a^{u})^{+}}
\global\long\def\pr{\Pr(a^{u})}

\title{Energy decay for the damped wave equation under a pressure condition}

\author{Emmanuel Schenck}

\email{emmanuel.schenck@cea.fr}
\begin{abstract}
We establish the presence of a spectral gap near the real axis for
the damped wave equation on a manifold with negative curvature. This
results holds under a dynamical condition expressed by the negativity
of a topological pressure with respect to the geodesic flow. As an
application, we show an exponential decay of the energy for all initial
data sufficiently regular. This decay is governed by the imaginary
part of a finite number of eigenvalues close to the real axis.
\end{abstract}
\maketitle

\section{Introduction}

One of the standard questions in geometric control theory concerns
the so-called stabilization problem: given a dissipative wave equation
on a manifold, one is interested in the behaviour of the solutions
and their energies for long times. The answers that can be given to
this problem are closely related to the underlying manifold and the
geometry of the control (or damping) region. 

In this paper, we shall study these questions in the particular case
of the damped wave equation on a compact Riemannian manifold $(M,g)$
with negative curvature and dimension $d\geq2$. For simplicity, we
will assume that $M$ has no boundary. If $a\in C^{\infty}(M)$ is
a real valued function on $M$, this equation reads\begin{equation}
(\d_{t}^{2}-\Delta_{g}+2a(x)\d_{t})u=0\,,\ \ (t,x)\in\R\times M\,,\label{eq:DWE}\end{equation}
 with initial conditions\begin{eqnarray*}
u(0,x) & = & \omega_{0}(x)\in H^{1}\\
\i\d_{t}u(0,x) & = & \omega_{1}(x)\in H^{0}.\end{eqnarray*}
Here $H^{s}\equiv H^{s}(M)$ are the usual Sobolev spaces on $M$.
The Laplace-Beltrami operator $\Dg\equiv\Delta$ is expressed in local
coordinates by \begin{equation}
\Dg=\frac{1}{\sqrt{\bar{g}}}\d_{i}(g^{ij}\sqrt{\bar{g}}\d_{j})\,,\quad\bar{g}=\det g.\label{eq:laplacian}\end{equation}
We will also denote by $\dvol=\sqrt{\bar{g}}dx$ the natural Riemannian
density, and $\scal uv=\int_{M}u\bar{v}\dvol$ the associated scalar
product. 

In all the following, we will consider only the case where the waves
are \textit{damped}, wich corresponds to take $a\geq0$ with $a$
non identically 0. We can reformulate the above problem into an equivalent
one by considering the unbounded operator \[
\cB=\left(\begin{array}{cc}
0 & \Id\\
-\Delta_{g} & -2\i a\end{array}\right)\,:H^{1}\times H^{0}\to H^{1}\times H^{0}\]
with domain $D(\cB)=H^{2}\times H^{1}$, and the following evolution
equation : \begin{equation}
(\d_{t}+\i\cB)\bold u=0\,,\ \ \bold u=(u_{0},u_{1})\in H^{1}\times H^{0}.\label{eq:DWE 2x2}\end{equation}
From the Hille-Yosida theorem, one can show that $\cB$ generates
a uniformly bounded, strongly continuous semigroup $\e^{-\i t\cB}$
for $t\geq0$, mapping any $(u_{0},u_{1})\in H^{1}\times H^{0}$ to
a solution $(u(t,x),\i\d_{t}u(t,x))$ of \eqref{eq:DWE 2x2}. Since
$\cB$ has compact resolvent, its spectrum $\Spec\cB$ consist in
a discrete sequence of eigenvalues $\{\tau_{n}\}_{n\in\N}$ . The
eigenspaces $E_{n}$ corresponding to the eigenvalues $\tau_{n}$
are all finite dimensional, and the sum $\bigoplus_{n}E_{n}$ is dense
in $H^{1}\times H^{0}$, see \cite{GK}. If $\tau\in\Spec\cB$, there
is $v\in H^{1}$ such that \begin{equation}
u(t,x)=\e^{-\i t\tau}v(x)\,,\label{eq:Stationary State}\end{equation}
and the function $u$ then satisfies \begin{equation}
P(\tau)u=0\,,\quad\trm{where}\quad P(\tau)=-\Delta-\tau^{2}-2\i a\tau.\label{eq:DWE-modes}\end{equation}
From \eqref{eq:DWE-modes}, it can be shown that the spectrum is symmetric
with respect to the imaginary axis, and satisfies \[
-2\|a\|_{\infty}\leq\Im\tau_{n}\leq0\]
while $|\Re\tau_{n}|\to\infty$ as $n\to\infty$ . Furthermore, if
$\Re\tau\neq0$, we have $\Im\tau\in[-\|a\|_{\infty},0]$, and the
only real eigenvalue is $\tau=0$, associated to the constant solutions
of \eqref{eq:DWE}. 

The question of an asymptotic density of modes has been adressed by
Markus and Matsaev in \cite{MM}, where they proved the following
Weyl-type law, also found later independently by Sjöstrand in \cite{Sjo}
:\[
\Card\{n:0\leq\Re\tau_{n}\leq\lambda\}=\lp\frac{\lambda}{2\pi}\rp^{d}\int_{p^{-1}([0,1])}dxd\xi+\cO(\lambda^{d-1})\,.\]
Here $p=g_{x}(\xi,\xi){}^{2}$ is the principal symbol of $-\Delta_{g}$
and $dxd\xi$ denotes the Liouville measure on $T^{*}M$ coming from
its symplectic structure. Under the asumption of ergodicity for the
geodesic flow with respect to the Liouville measure, Sjöstrand also
showed that most of the eigenvalues concentrate on a line in the high-frequency
limit. More precisely, he proved that given any $\eps>0$, \begin{equation}
\Card\{n:\tau_{n}\in[\lambda,\lambda+1]+\i(\R\setminus[-\bar{a}-\eps,-\bar{a}+\eps])\}=o(\lambda^{d-1})\:.\label{eq:Accumulation}\end{equation}
The real number $\bar{a}$ is the ergodic mean of $a$ on the unit
cotangent bundle $S^{*}M=\{(x,\xi)\in T^{*}M,g_{x}(\xi,\xi)=1\}$.
It is given by \[
\bar{a}=\lim_{T\to\infty}T^{-1}\int_{0}^{T}a\circ\Phi^{t}dt,\ \trm{well\ defined\ }dxd\xi-\trm{almost\ everywhere\ on\ }S^{*}M\:.\]
 Hence the eigenvalues close to the real axis, say with imaginary
parts in $[\alpha,0]$, $0>\alpha>-\bar{a}$ can be considered as
{}``exceptional''. The first result we will present in this paper
show that a spectral gap of finite width can actually exist below
the real axis under some dynamical hypotheses, see Theorem \ref{thm:gapW}
below. 

The second object studied in this work is the energy of the waves.
From now on, we call $\cH=H^{0}\times H^{1}$ the space of Cauchy
data. Let $u$ be a solution of \eqref{eq:DWE} with initial data
$\bs\omega\in\cH$. The energy of $u$ is defined by \[
E(u,t)=\fr12(\|\d_{t}u\|_{L^{2}}^{2}+\|\nabla u\|_{L^{2}}^{2})\,.\]
As a well known fact, $E$ is decreasing in time, and $E(u,t)\stackrel{t\to\infty}{\longrightarrow}0$.
It is then natural to ask if a particular rate of decay of the energy
can be identified. Let $s>0$ be a positive number, and define the
Hibert space \[
\cH^{s}=H^{1+s}\times H^{s}\subset\cH\,.\]
Generalizing slightly a definition of Lebeau, we introduce the best
exponential rate of decay with respect to $\|\cdot\|_{\cH^{s}}$ as
\begin{equation}
\rho(s)=\sup\{\beta\in\R_{+}:\exists C>0\ \trm{such\ that\ }\forall\bs\omega\in\cH^{s},\ E(u,t)\leq C\e^{-\beta t}\|\bs\omega\|_{\cH^{s}}\}\label{eq:Best decay rate}\end{equation}
where the solutions $u$ of \eqref{eq:DWE} have been identified with
the Cauchy data $\bs\omega\in\cH^{s}$. It is shown in \cite{Leb}
that \[
\rho(0)=2\min(G,C(\infty)),\]
where $G=\inf\lbrace-\Im\tau;\tau\in\Spec\cB\setminus\{0\}\rbrace$
is the spectral gap, and \[
C(\infty)=\lim_{t\to\infty}\inf_{\rho\in T^{*}M}\frac{1}{t}\int_{0}^{t}\pi^{*}a(\Phi^{s}\rho)ds\geq0.\]
Here $\Phi^{t}:T^{*}M\to T^{*}M$ is the geodesic flow, and $\pi:T^{*}M\to M$
is the canonical projection along the fibers. It follows that the
presence of a spectral gap below the real axis is of significative
importance in the study of the energy decay. However, an explicit
example is given in \cite{Leb}, where $G>0$ while $C(\infty)=0$,
and then $\rho(0)=0$ . This particular situation is due to the failure
of the geometrical control, namely, the existence of orbits of the
geodesic flow not meeting $\supp a$ (which implies $C(\infty)=0$).
Hence, the spectrum of $\cB$ may not always control the energy decay,
and some dynamical assumptions on the geodesic flow are required if
we want to solve positively the stabilization problem. In the case
where geometric control holds \cite{RT}, it has been shown in various
settings that $\rho(0)>0$, see for instance \cite{BLR,Leb,Hit}.
In \cite{Chr}, a particular situation is analyzed where the geometric
control does not hold near a closed hyperbolic orbit of the geodesic
flow: in this case, there is a sub-exponential decay of the energy
with respect to $\|\cdot\|_{\cH^{\eps}}$ for some $\eps>0$. 

\medskip{}

\textbf{Dynamical assumptions.} In this paper, we first assume $(M,g)$
has strictly negative sectional curvatures. This implies that the
geodesic flow has the Anosov property on every energy layer, see Section
\ref{sec:flows} below. Without loss of generality, we suppose that
the injectivity radius satisfies $r\geq2$. Then, we drop the geometric
control assumption, and replace it with a dynamical hypothese involving
the topological pressure of the geodesic flow on $S^{*}M$, which
we define now. For ervery $\eps>0$ and $T>0$, a set $\cS\subset S^{*}M$
is $(\eps,T)-$separated if $\rho,\theta\in\cS$ implies that $d(\Phi^{t}\rho,\Phi^{t}\theta)>\eps$
for some $t\in[0,T]$, where $d$ is the distance induced from the
adapted metric on $T^{*}M$. For $f$ continuous on $S^{*}M$, set
\[
Z(f,T,\eps)=\sup_{\cS}\left\{ \sum_{\rho\in\cS}\exp\sum_{k=0}^{T-1}f\circ\Phi^{k}(x)\right\} .\]
The topological pressure $\Pr(f)$ of the function $f$ with respect
to the geodesic flow is defined by \[
\Pr(f)=\lim_{\eps\to0}\limsup_{T\to\infty}\frac{1}{T}\log Z(f,T,\eps).\]
The pressure $\Pr(f)$ contains useful information on the Birkhoff
averages of $f$ and the complexity of the geodesic flow, see for
instance \cite{Wal} for a general introduction and further properties.
The particular function we will deal with is given by \begin{equation}
a^{u}:\rho\in S^{*}M\mapsto a^{u}(\rho)=-\int_{0}^{1}\pi^{*}a\circ\Phi^{s}(\rho)\: ds+\fr12\log J^{u}(\rho)\in\R\label{eq:Fonction Pression}\end{equation}
where $J^{u}(\rho)$ is the unstable Jacobian at $\rho$ for time
$1$, see Section \ref{sec:flows}. In this paper, we will always
assume that \begin{equation}
\Pr(a^{u})<0.\label{eq:Condition Pression}\end{equation}

\medskip{}
\textbf{Main results. }Under the condition $\Pr(a^{u})<0$, we will
see that a spectral gap of finite width exists below the real axis.
As a consequence, there is an exponential decay of the energy of the
waves with respect to $\|\cdot\|_{\cH^{\ka}}$ for any $\ka\geq d/2$,
and if $G<|\Pr(a^{u})|$, we have $\rho(\ka)=2G$. We begin by stating
the result concerning the spectral gap. 
\begin{thm}
\textbf{\textup{(Spectral gap) }}\label{thm:gapW} Suppose that the
topological pressure of $a^{u}$ with respect to the geodesic flow
on $S^{*}M$ satisfies $\Pr(a^{u})<0$, and let $\eps>0$ be such
that \[
\Pr(a^{u})+\eps<0.\]
Then, there exisits $e_{0}(\eps)>0$ such that for any $\tau\in\Spec\cB$
with $|\Re\tau|\geq e_{0}(\eps)$, we have \[
\Im\tau\leq\Pr(a^{u})+\eps.\]

\end{thm}
The presence of a spectral gap of finite width below the real axis
is not obvious a priori if geometric control does not hold, since
there may be a possibility for $|\Im\tau_{n}|$ to become arbitrary
small as $n\to\infty$ : see for instance \cite{Hit}, Theorem 1.3.
However, this accumulation on the real axis can not occur faster than
a fixed exponential rate, as it was shown in \cite{Leb} that \[
\exists C>0\ \trm{such\ that}\ \forall\tau\in\Spec\cB\,,\ \ \Im\tau\leq-\frac{1}{C}\e^{-C|\Re\tau|}.\]

Let us mention a result comparable to Theorem \ref{thm:gapW} in the
framework of chaotic scattering obtained recently by Nonnenmacher
and Zworski \cite{NZ}, in the semiclassical setting. For a large
class of Hamiltonians, including $P(\h)=-\h\Delta+V$ on $\R^{d}$
with $V$ compactly supported, they were able to show a resonance-free
region near the energy $E$:\[
\exists\delta,\gamma>0\ \trm{such\ that\ }\Res(P(\h))\cap([E-\delta,E+\delta]-\i[0,\gamma\h])=\emptyset\ \trm{for\ }0<\h\leq\h_{\delta,\gamma}.\]
 This holds provided that the hamiltonian flow $\Phi^{t}$ on the
trapped set $K_{E}$ at energy $E$ is hyperbolic, and that the pressure
of the unstable Jacobian with respect to the geodesic flow on $K_{E}$
is strictly negative. We will adapt several techniques of \cite{NZ}
to prove Theorem \ref{thm:gapW}, some of them coming back to \cite{Ana1,AN}.

In a recent paper, Anantharaman \cite{Ana2} studied the spectral
deviations of $\Spec\mbox{ }\cB$ with respect to the line of accumulation
$\Im z=-\bar{a}$ appearing in \eqref{eq:Accumulation}. In the case
of constant negative curvature, she obtained an upper bound for the
number of modes with imaginary parts above $-\bar{a}$ , and showed
that for $\alpha\in[-\bar{a},0[$, there exists a function $H(\alpha)$
such that \begin{equation}
\forall c>0,\forall\eps>0,\ \ \limsup_{\lambda\to\infty}\frac{\log\Card\{\tau_{n}:\Re\tau_{n}\in[\lambda-c,\lambda+c],\ \Im\tau_{n}\geq\alpha+\eps\}}{\log\lambda}\leq H(\alpha).\label{eq:Deviations}\end{equation}
$H(\alpha)$ is a dynamical quantity defined by \[
H(\alpha)=\sup\{h_{KS}(\mu),\ \mu\in\cM_{\fr12},\ \int ad\mu=-\alpha\}\]
where $\cM_{\fr12}$ denotes the set of $\Phi^{t}-$invariant measures
on $S^{*}M$, and $h_{KS}$ stands for the Kolmogorov--Sinai entropy
of $\mu$. As a consequence of Theorem \ref{thm:gapW}, the result
of Anantharaman is not always optimal : $H(\alpha)\neq0$ for $\alpha\in[-\bar{a},0[$,
but if $\Pr(a^{u})<0$, there is no spectrum in a strip of finite
width below the real axis, i.e. the $\limsup$ in \eqref{eq:Deviations}
vanishes for some $\alpha=\alpha(a)\neq0$. 

The operator $\cB$ being non-selfadjoint, its eigenfunctions may
fail to form a Riesz basis of $\cH$. However, if a solution $\bs u$
of \eqref{eq:DWE 2x2} has initial data sufficiently regular, it is
still possible to expand it on eigenfunctions which  eigenmodes are
close to the real axis, up to an exponentially small error in time
: 
\begin{thm}
\textbf{\textup{(Eigenvalues expansion)}}\label{thm:Eigenvalues expansion}
Let $\eps>0$ such that $\Pr(a^{u})+\eps<0$, and $\ka\geq\frac{d}{2}$.
There exists $e_{0}(\eps)>0$, $n=n(\eps)\in\N$ and a (finite) sequence
$\tau_{0},\dots,\tau_{n-1}$ of eigenvalues of $\cB$ with \[
\tau_{j}\in[-e_{0}(\eps),e_{0}(\eps)]+\i[\pr+\eps,0],\quad j\in\llb0,n-1\rrb\,,\]
such that for any solution $\bs u(t,x)$ of \eqref{eq:DWE 2x2} with
initial data $\bs\omega\in\cH^{\ka}$, we have \[
\bs u(t,x)=\sum_{j=0}^{n-1}\e^{-\i t\tau_{j}}\bs u_{j}(t,x)+\bold r_{n}(t,x)\,,\ \ t>0\,.\]
The functions $\bs u_{j},r_{n}$ satisfy \[
\|\bs u_{j}(t,\cdot)\|_{\cH}\leq Ct^{m_{j}}\|\bs\omega\|_{\cH}\ \ \et\ \ \|\bold r_{n}(t,\cdot)\|_{\cH}\leq C_{\eps}\e^{t(\Pr(a^{u})+\eps)}\|\bs\omega\|_{\cH^{\ka}},\]
where $m_{j}$ denotes the multiplicity of $\tau_{j}$, the constants
$C>0$ depends only on $M$ and $a$, while $C_{\eps}>0$ depending
on $M,a$ and $\eps$. 
\end{thm}
A similar eigenvalues expansion can be found in \cite{Hit}, where
no particular assumption on the curvature of $M$ is made, however
the geometric control must hold. Our last result deals with an exponential
decay of the energy, which will be derived a consequence of the preceding
theorem : 
\begin{thm}
\textbf{\textup{(Exponential energy decay)}}\label{thm:energy} Let
$\eps>0$, $(\tau_{j})_{0\leq j\leq n(\eps)-1}$, $\ka$ and $\bs u$
as in Theorem \ref{thm:Eigenvalues expansion}. Set by convention
$\tau_{0}=0$. The energy $E(u,t)$ satisfies \[
E(u,t)\leq\lp\sum_{j=1}^{n-1}\e^{t\Im\tau_{j}}t^{m_{j}}C\|\bs\omega\|_{\cH}+C_{\eps}\e^{t(\Pr(a^{u})+\eps)}\|\bs\omega\|_{\cH^{\ka}}\rp^{2}\]
where $m_{j}$ denotes the multiplicity of $\tau_{j}$. The constants
$C>0$ depends only on $M$ and $a$, while $C_{\eps}>0$ depending
on $M,a$ and $\eps$. In particular, $\rho(\ka)=2\min(G,|\Pr(a^{u})+\eps|)>0$. \end{thm}
\begin{rem*}
In our setting, it may happen that geometric control does not hold,
while $\Pr(a^{u})<0$. In this particular situation, it follows from
\cite{BLR} that we can not have an exponential energy decay uniformly
for all Cauchy data in $\cH$, where by uniform we mean that the constant
$C$ appearing in \eqref{eq:Best decay rate} does not depend on $\bs u$.
However, if for $\ka\geq\frac{d}{2}$ we look at $\rho(\ka)$ instead
of $\rho(0)$, our results show that we still have uniform exponential
decay, namely $\rho(\ka)>0$ while $\rho(0)=0$.
\end{rem*}

\subsection{Semiclassical reduction}

The main step yielding to Theorem \eqref{thm:gapW} is more easily
achieved when working in a semiclassical setting. From the eigenvalue
equation \eqref{eq:DWE-modes}, we are lead to study the equation
\[
P(\tau)u=0\]
where $\Im\tau=\cO(1)$. To obtain a spectral gap below the real axis,
we are lead to study eigenvalues with arbitrary large real parts since
$\Spec\cB$ is discrete. For this purpose, we introduce a semiclassical
parameter $\h\in]0,1]$, and write the eigenvalues as \[
\tau=\frac{1}{\h}+\cO(1).\]
If we let $\h$ go to 0, the eigenvalues $\tau$ we are interested
in then satisfy $\tau\h\hto01$ . Putting $\tau=\frac{\lambda}{\h}$
and $z=\lambda^{2}/2$, we rewrite the stationary equation \[
\lp-\frac{\h^{2}\Delta}{2}-z-\i\h q_{z}\rp u=0\,,\quad q_{z}(x)=\sqrt{2z}a(x).\]
Equivalently, we write\begin{equation}
(\cP(z,\h)-z)u=0\label{eq:eigenmodes}\end{equation}
 where $\cP(z,\h)=-\frac{\h^{2}\Delta}{2}-\i\h q_{z}$. The parameter
$z$ plays the role of a complex eigenvalue of the non-selfadjoint
quantum Hamiltonian $\cP.$ It is close to the {}``energy'' $E=1/2$,
while $\Im z$  is of order $\h$ and represents the {}``decay rate''
of the mode. In order to recall these properties, we will often write
\begin{equation}
z=\fr12+\h\zeta\,,\ \ \zeta\in\C\ \et\ |\zeta|=\cO(1).\label{eq: z - zeta}\end{equation}
In most of the following, we will deal with the semiclassical analysis
of the non-selfadjoint Schrödinger operator $\cP(z,\h)$ and the associated
Schrödinger equation \begin{equation}
\i\h\d_{t}\Psi=\cP(z,\h)\Psi\,\quad\trm{with}\ \ \|\Psi\|_{L^{2}}=1.\label{eq:Schrodinger}\end{equation}
The basic facts and notations we will use from semiclassical analysis
are recalled in Appendix \ref{sec:semiclass}. The operator $\cP$
has a principal symbol equal to $p(x,\xi)=\fr12g_{x}(\xi,\xi)$ ,
and a subprincipal symbol given by $-\i q_{z}$. Note that the classical
Hamiltonian $p(x,\xi)$ generates the geodesic flow on the energy
surface $p^{-1}(\fr12)=S^{*}M$. The properties of the geodesic flow
on $S^{*}M$ which will be useful to us are summarized in the next
section, where is also given an alternative definition of the topological
pressure more adapted to our purposes. We will denote the quantum
propagator by \[
\cU^{t}\equiv\e^{-\frac{\i t}{\h}\cP},\]
so that if $\Psi\in L^{2}(M)$ satisfies \eqref{eq:Schrodinger},
we have $\Psi(t)=\cU^{t}\Psi(0).$ Using standard methods of semiclassical
analysis, one can show that $\cU^{t}$ is a Fourier integral operator
(see \cite{EZ}, chapter 10) associated with the symplectic diffeomorphism
given by the geodesic flow $\Phi^{t}$. Since we assumed that $a\geq0$,
it is true that $\|\cU^{t}\|_{L^{2}\to L^{2}}\leq1,\ \forall t\geq0$. 

Denote \[
\Sigma_{\fr12}=\{z=\fr12+\cO(\h)\in\C,\ \exists\Psi\in L^{2}(M),\ (\cP(z,\h)-z)\Psi=0\}.\]
If $z\in\Sigma_{\fr12}$ and $\Psi$ is such that \eqref{eq:eigenmodes}
holds, the semiclassical wave front set of $\Psi$ satisfies \[
\WF_{\h}(\Psi)\subset S^{*}M.\]
This comes from the fact that $\Psi$ is an eigenfunction associated
with the eigenvalue $\fr12$ of a pseudodifferential operator with
principal symbol $p(x,\xi)=\fr12g_{x}(\xi,\xi)$. Using these semiclassical
settings, we will show the following key result : 
\begin{thm}
\label{thm:gap}Let $z\in\Sigma_{\fr12}$ , and $\eps>0$ be such
that $\Pr(a^{u})+\eps<0$. There exists $\h_{0}=\h_{0}(\eps)$ such
that \[
\h\leq\h_{0}\ \Rightarrow\ \frac{\Im z}{\h}\leq\Pr(a^{u})+\eps.\]

\end{thm}
From \eqref{eq: z - zeta}, we also notice that the above equation
implies $\Im\tau\leq\Pr(a^{u})+\eps+\cO(\h)$ since $\tau=\h^{-1}\sqrt{2z}$,
and then $\Im\tau=\Im\zeta+\cO(\h)$. It follows by rescaling that
Theorem \ref{thm:gap} is equivalent to Theorem \ref{thm:gapW}.

\section{Quantum dynamics and spectral gap\label{sec:Spectral Gap}}

\subsection{Hyperbolic flow and topological pressure\label{sec:flows}}

We call \[
\Phi^{t}=\e^{tH_{p}}:T^{*}M\to T^{*}M\]
 the geodesic flow, where $H_{p}$ is the Hamilton vector field of
$p$. In local coordinates, \[
H_{p}\defeq\sum_{i=1}^{d}\frac{\d p}{\d\xi_{i}}\d_{x_{i}}-\frac{\d p}{\d x_{i}}\d_{\xi_{i}}=\{p,\cdot\}\]
where the last equality refers to the Poisson bracket with respect
to the canonical symplectic form $\omega=\sum_{i=1}^{d}d\xi_{i}\wedge dx_{i}$.
Since $M$ has strictly negative curvature, the flow generated by
$H_{p}$ on constant energy layers $\cE=p^{-1}(E)\subset T^{*}M,\ E>0$
has the Anosov property: for any $\rho\in\cE$, the tangent space
$T_{\rho}\cE$ splits into flow, stable and unstable subspaces\[
T_{\rho}\cE=\R H_{p}\oplus E^{s}(\rho)\oplus E^{u}(\rho)\,.\]
The spaces $E^{s}(\rho)$ and $E^{u}(\rho)$ are $d-1$ dimensional,
and are stable under the flow map: \[
\forall t\in\R,\ \ d\Phi_{\rho}^{t}(E^{s}(\rho))=E^{s}(\Phi^{t}(\rho)),\quad d\Phi_{\rho}^{t}(E^{u}(\rho))=E^{u}(\Phi^{t}(\rho)).\]
Moreover, there exist $C,\lambda>0$ such that \begin{eqnarray}
i) &  & \|d\Phi_{\rho}^{t}(v)\|\leq C\e^{-\lambda t}\|v\|,\ \trm{\ for\ all\ }v\in E^{s}(\rho),\ t\geq0\nonumber \\
ii) &  & \|d\Phi_{\rho}^{-t}(v)\|\leq C\e^{-\lambda t}\|v\|,\ \trm{\ for\ all\ }v\in E^{u}(\rho),\ t\geq0.\label{eq:Unstable}\end{eqnarray}
One can show that there exist a metric on $T^{*}M$ call the adapted
metric, for which one can takes $C=1$ in the preceding equations.
At each point $\rho$, the spaces $E^{u}(\rho)$ are tangent to the
unstable manifold $W^{u}(\rho)$, the set of points $\rho^{u}\in\cE$
such that $d(\Phi^{t}(\rho^{u}),\Phi^{t}(\rho))\xrightarrow{t\to-\infty}0$
where $d$ is the distance induced from the adapted metric. Similarly,
$E^{s}(\rho)$ is tangent to the stable manifold $W^{s}(\rho)$, the
set of points $\rho^{s}$ such that $d(\Phi^{t}(\rho^{s}),\Phi^{t}(\rho))\xrightarrow{t\to+\infty}0$. 

The adapted metric induces a the volum form $\Omega_{\rho}$ on any
$d$ dimensional subspace of $T(T_{\rho}^{*}M)$. Using $\Omega_{\rho}$,
we now define the unstable Jacobian at $\rho$ for time $t$. Let
us define the weak-stable and weak-unstable subspaces at $\rho$ by
\[
E^{s,0}(\rho)=E^{s}(\rho)\oplus\R H_{p}\,,\quad E^{u,0}(\rho)=E^{u}(\rho)\oplus\R H_{p}.\]
We set \[
J_{t}^{u}(\rho)=\det d\Phi^{-t}|_{E^{u,0}(\Phi^{t}(\rho))}=\frac{\Omega_{\rho}(d\Phi^{-T}v_{1}\wedge\dots\wedge d\Phi^{-t}v_{d})}{\Omega_{\Phi^{t}(\rho)}(v_{1}\wedge\dots\wedge v_{d})}\,,\ \ \ \ J^{u}(\rho)\defeq J_{1}^{u}(\rho),\]
where $(v_{1},\dots,v_{d})$ can be any basis of $E^{u,0}(\rho)$.
While we do not necessarily have $J^{u}(\rho)<1$, it is true that
$J_{t}^{u}(\rho)$ decays exponentially as $t\to+\infty$. 

The definition of the topological pressure of the geodesic flow given
in the introduction, although quite straighforward to state, is not
really suitable for our purposes. The alternative definition of the
pressure we will work with is based on refined covers of $S^{*}M$,
and can be stated as follows. For $\delta>0$, let $\cE^{\delta}=p^{-1}[\fr12-\delta,\fr12+\delta]$
be a thin neighbourhood of the constant energy surface $p^{-1}(\fr12)$,
and $V=\{V_{\alpha}\}_{\alpha\in I}$ an open cover of $\cE^{\delta}$.
In what follows, we shall always choose $\delta<1/2$. For $T\in\N^{*}$,
we define the refined cover $V^{(T)}$, made of the sets \[
V_{\beta}=\bigcap_{k=0}^{T-1}\Phi^{-k}(V_{b_{k}})\,,\quad\beta=b_{0}b_{1}\dots b_{T-1}\in I^{T}.\]
 It will be useful to coarse-grain any continuous function $f$ on
$\cE^{\delta}$ with respect to $V^{(T)}$ by setting \[
\moy f_{T,\beta}=\sup_{\rho\in V_{\beta}}\sum_{i=0}^{T-1}f\circ\Phi^{i}(\rho).\]
One then define \[
Z_{T}(V,f)=\inf_{B_{T}}\left\{ \sum_{\beta\in B_{T}}\exp(\moy f_{T,\beta})\ :\ B_{T}\subset I^{T}\,,\ \cE^{\delta}\subset\bigcup_{\beta\in B_{T}}V_{\beta}\right\} \,.\]
The topological pressure of $f$ with respect to the geodesic flow
on $\cE^{\delta}$ is defined by :\[
\press^{\delta}(f)=\lim_{\diam V\to0}\lim_{T\to\infty}\frac{1}{T}\log Z_{T}(V,f)\,.\]
 The pressure on the unit tangent bundle $S^{*}M$ is simply obtained
by continuity, taking the limit $\Pr(f)=\lim_{\delta\to0}\Pr^{\delta}(f)$.
To make the above limits easier to work with, we now take $f=a^{u}$
and fix $\eps>0$ such that $\Pr(a^{u})+\eps<0$. Then, we choose
the width of the energy layer $\delta\in]0,1[$ sufficiently small
such that $|\Pr(a^{u})-\press^{\delta}(a^{u})|\leq\eps/2$. Given
a cover $\cV=\{\cV_{\alpha}\}_{\alpha\in\cA}$ (of arbitrary small
diameter), there exist a time $t_{0}$ depending on the cover $\cV$
such that \[
\left|\frac{1}{t_{0}}\log Z_{t_{0}}(\cV,a^{u})-\press^{\delta}(a^{u})\right|\leq\frac{\eps}{2}\,.\]
Hence there is a subset of $t_{0}-$strings $\cB_{t_{0}}\subset\cA^{t_{0}}$
such that $\{\cV_{\alpha}\}_{\alpha\in\cB_{t_{0}}}$ is an open cover
of $\cE^{\delta}$ and satisfies \begin{equation}
\sum_{\beta\in\cB_{t_{0}}}\exp(\moy{a^{u}}_{t_{0},\beta})\leq\exp\lp t_{0}(\press^{\delta}(a^{u})+\frac{\eps}{2})\rp\leq\exp\lp t_{0}(\press(a^{u})+\eps)\rp\,.\label{eq:PressApprox}\end{equation}
For convenience, we denote by $\{\cW_{\beta}\}_{\beta\in\cB_{t_{0}}}\equiv\{\cV_{\beta}\}_{\beta\in\cB_{t_{0}}}$
the sub-cover of $\cV^{(t_{0})}$ such that \eqref{eq:PressApprox}
holds. Note that in this case, the diameter of $\cV$, $t_{0}$ and
then $\cW$ depends on $\eps$.

\subsection{Discrete time evolution\label{sub:MicroPartition}}

Let $\{\varphi_{\beta}\}_{\beta\in\cB_{t_{0}}}$ be a partition of
unity adapted to $\cW$, so that its Weyl quantization $\varphi_{\beta}^{w}\defeq\Pi_{\beta}$
(see Appendix \ref{sec:semiclass}) satisfy \[
\WF_{\h}(\Pi_{\beta})\subset\cE^{\delta},\quad\Pi_{\beta}^{*}=\Pi_{\beta},\quad\sum_{\beta}\Pi_{\beta}=\bbbone\quad\trm{microlocally\ near\ }\cE^{\delta/2}\,.\]
We will also consider a partition of unity $\{\ti\varphi_{\alpha}\}_{\alpha\in\cA}$
adapted to the cover $\cV,$ and its Weyl quantization $\ti\Pi\defeq\ti\varphi^{w}$.
In what follows, we will be interested in the propagator $\cU^{Nt_{0}+1}$
, and \[
N=T\log\h^{-1},\ T>0\,.\]
It is important to note that $T$ can be arbitrary large, but is fixed
with respect to $\h.$ The propagator $\cU^{Nt_{0}}$ is decomposed
by inserting $\sum_{\beta\in\cB_{t_{0}}}\Pi_{\beta}$ at each time
step of length $t_{0}$. Setting first $\cU_{\beta}=\cU^{t_{0}}\Pi_{\beta}$,
we have (microlocally near $\cE^{\delta/2}$) the equality $\cU^{t_{0}}=\sum_{\beta}\cU_{\beta}$,
and then \begin{equation}
\cU^{Nt_{0}}=\sum_{\beta_{1},\beta_{2},\dots,\beta_{N}\in\cB_{t_{0}}}\cU_{\beta_{N}}\dots\cU_{\beta_{1}}\,,\quad\trm{near}\ \cE^{\delta/2}.\label{eq:DecompPropag}\end{equation}

\subsection{Proof of Theorem \ref{thm:gap}}

We begin by choosing $\chi\in C_{0}^{\infty}(T^{*}M)$ such that $\supp\chi\Subset\cE^{\delta}$
and $\chi\equiv1$ on $\cE^{\delta/4}$, and considering $\oph(\chi)$.
Applying the Cauchy-Schwartz inequality, we get immediately 

\begin{equation}
\|\cU^{Nt_{0}+1}\oph(\chi)\|\leq\sum_{\beta_{1},\beta_{2},\dots,\beta_{N}\in\cB_{t_{0}}^{N}}\|\cU_{\beta_{N}}\dots\cU_{\beta_{1}}\cU^{1}\oph(\chi)\|+\cO_{L^{2}\to L^{2}}(\h^{\infty}).\label{eq:Insertion Identity}\end{equation}
Unless otherwise stated, the norms $\|\cdot\|$ always refer to $\|\cdot\|_{L^{2}\to L^{2}}$
or $\|\cdot\|_{L^{2}}$, according to the context. The proof of Theorem
\eqref{thm:gap} relies on the following intermediate result, proven
much later in Section \ref{sec:HypDispEst}. 
\begin{prop}
\label{pro:HypDispEst}\textbf{\textup{(Hyperbolic dispersion estimate)}}
Let $\eps>0$, and $\delta,\cV,t_{0}$ be as in Section \ref{sec:flows}.
For $N=Tt_{0}\log\h^{-1}$, $T>0$, take a sequence $\beta_{1},\dots\beta_{N}$
and $\cW_{1},\dots,\cW_{\beta_{N}}$ the associated open sets of the
refined cover $\cW$. Finally , let $\oph(\chi)$ be as above. There
exists a constant $C>0$ and $\h_{0}(\eps)\in]0,1[$ such that \[
\h\leq\h_{0}\ \Rightarrow\ \|\cU^{t_{0}}\Pi_{\beta_{N}}\dots\cU^{t_{0}}\Pi_{\beta_{1}}\cU^{1}\oph(\chi)\|\leq C\h^{-d/2}\prod_{j=1}^{N}\e^{\moy{a^{u}}_{t_{0},\beta_{j}}}\]
where $\moy{a^{u}}_{t_{0},\beta}=\sup_{\rho\in\cW_{\beta}}\sum_{j=0}^{t_{0}-1}a^{u}\circ\Phi^{j}(\rho).$
The constant $C$ only depends on the manifold $M$.
\end{prop}
We also state the following crucial consequence :
\begin{cor}
\label{cor:Decay Propag} Take $\eps>0$ such that $\pr+\eps<0$.
There exists $C>0$ and $\h_{0}(\eps)\in]0,1[$ such that \[
\h\leq\h_{0}\Rightarrow\|\cU^{Nt_{0}+1}\oph(\chi)\|\leq C\h^{-\frac{d}{2}}\e^{Nt_{0}(\Pr(a^{u})+\eps)}\]
The constant $C$ only depends on $M$. \end{cor}
\begin{proof}
Given $\eps>0$, we choose $\delta,\cV,t_{0},\cW$ as in the preceding
proposition. Using \eqref{eq:Insertion Identity}, we then have \begin{eqnarray*}
\|\cU^{Nt_{0}+1}\oph(\chi)\| & \leq & C\h^{-\frac{d}{2}}\sum_{\beta_{1}\dots\beta_{N}\in\cB_{t_{0}}^{N}}\lp\prod_{j=1}^{N}\e^{\moy{a^{u}}_{t_{0},\beta_{j}}}+\cO(\h^{\infty})\rp\\
 & \leq & C\h^{-\frac{d}{2}}\lp\sum_{\beta\in\cB_{t_{0}}}\e^{\moy{a^{u}}_{t_{0},\beta}}\rp^{N}+\cO(\h^{\infty}).\end{eqnarray*}
To get the second line, notice that the number of terms in the sum
is of order $(\Card\cB_{t_{0}})^{N}=\h^{-T\log\Card\cB_{t_{0}}}$.
From our choice of $\eps$ and $\delta$, we can use \eqref{eq:PressApprox},
and for $\h$ small enough, rewrite this equation as \begin{eqnarray*}
\|\cU^{Nt_{0}+1}\oph(\chi)\| & \leq & C\h^{-\frac{d}{2}}\e^{Nt_{0}(\Pr^{\delta}(a^{u})+\eps/2)}\\
 & \leq & C\h^{-\frac{d}{2}}\e^{Nt_{0}(\Pr(a^{u})+\eps)}\end{eqnarray*}
where $C>0$ only depends on the manifold $M$. 
\end{proof}
Let us show how this result implies Theorem \ref{thm:gap}. We assume
that $\Psi$ satisfies \eqref{eq:eigenmodes}, and therefore $\|\cU^{Nt_{0}+1}\Psi\|=\e^{\frac{(Nt_{0}+1)\Im z}{\h}}$.
Notice also that we have \[
\oph(\chi)\Psi=\Psi+\cO(\h^{\infty})\]
since $\WF_{\h}(\Psi)\subset S^{*}M$, and then \[
\|\cU^{Nt_{0}+1}\oph(\chi)\Psi\|=\|\cU^{Nt_{0}+1}\Psi\|+\cO(\h^{\infty})=\e^{(Nt_{0}+1)\Im z}+\cO(\h^{\infty})\,.\]
It follows from the corollary that \[
\e^{\frac{Nt_{0}+1}{\h}\Im z}\leq C\h^{-\frac{d}{2}}\e^{Nt_{0}(\Pr(q^{u})+\eps)}+\cO(\h^{m})\,,\]
where $m$ can be arbitrary large. Taking the logarithm, this yields
to \[
\frac{\Im z}{\h}\leq\frac{\log C}{Nt_{0}}-\frac{d}{2Nt_{0}}\log\h+\press(q^{u})+\eps+\cO(\frac{1}{Nt_{0}})\,.\]
But given $\eps>0$, we can take $N=T\log\h^{-1}$ with $T$ arbitrary.
Hence there is $\h_{0}(\eps)\in]0,1[$ and $T$ sufficiently large,
such that \[
\h\leq\h_{0}(\eps)\Rightarrow\frac{\Im z}{\h}\leq\press(q^{u})+2\eps.\]
Since the parameter $\eps$ can be chosen as small as wished, this
proves Theorem \ref{thm:gap}.

\section{Eigenvalues expansion and energy decay\label{sec:Energy}}

\subsection{Resolvent estimates}

To show the exponential decay of the energy, we follow a standard
route from resolvent estimates in a strip around the real axis. Let
us denote \[
Q(z,\h)=-\frac{\h^{2}}{2}\Delta-z-\i\h\sqrt{2z}a(x)=\cP(z,\h)-z.\]
The following proposition establish a resolvent estimate in a strip
of width $|\pr+\eps|$ below the real axis in the semiclassical limit.
This is the main step toward Theorems \ref{thm:Eigenvalues expansion}
and \ref{thm:energy}, see also \cite{NZ2} for comparable resolvent
estimates in the chaotic scattering situation : 
\begin{prop}
Let $\eps>0$. Choose $\ga<0$ such that \[
\Pr(a^{u})+\eps<\ga<0\,,\]
and $z=\fr12+\h\zeta,$ with $|\zeta|=\cO(1)$ satisfying \[
\gamma\leq\Im\zeta\leq0\,.\]
There exists $\h_{0}(\eps)>0$, $C_{\eps}>0$ depending on $M,a$
and $\eps$ such that \[
\h\leq\h_{0}(\eps)\ \Rightarrow\ \|Q(z,\h)^{-1}\|_{L^{2}\to L^{2}}\leq C_{\eps}\h^{-1+c_{0}\Im\zeta}\lh\]
where $c_{0}=\frac{d}{2|\Pr(q^{u})+\eps|}$. \end{prop}
\begin{proof}
Given $\eps>0$, we fix $\h_{0}(\eps)$ so that Corollary \ref{cor:Decay Propag}
holds. Finally, we define \[
\Pr(a^{u})^{+}=\Pr(a^{u})+\eps.\]
 In order to bound $Q(z,\h)^{-1}$, we proceed in two steps, by finding
two operators which approximate $Q^{-1}$: one on the energy surface
$\cE^{\delta}$, the other outside $\cE^{\delta}$. Let $\chi$ be
as in Section \ref{sec:Spectral Gap}, and choose also Let $\ti\chi\in C_{0}^{\infty}(T^{*}M)$
with $\supp\ti\chi\Subset\supp\chi$, such that we also have $\ti\chi=1$
near $S^{*}M$. We first look for an operator $A_{0}=A_{0}(z,\h)$
such that \[
QA_{0}=(1-\oph(\chi))+\cO_{L^{2}\to L^{2}}(\h^{\infty}).\]
For this, consider $Q(z,\h)+\i\oph(\ti\chi)\defeq Q_{0}(z,\h)$. Because
of the property of $\ti\chi$, the operator $Q_{0}$ is elliptic.
Hence, there is an operator $\ti A_{0}$, uniformly bounded in $L^{2}(M)$,
such that \[
Q_{0}\ti A_{0}=\Id+\cO_{L^{2}\to L^{2}}(\h^{\infty}).\]
The operator $A_{0}$ we are looking for is obtained by taking $A_{0}=\ti A_{0}(1-\oph(\chi))$.
Indeed, \begin{eqnarray*}
Q(z,\h)A_{0}(z,\h) & = & 1-\oph(\chi)-\i\oph(\ti\chi)\ti A_{0}(1-\oph(\chi))+\cO_{L^{2}\to L^{2}}(\h^{\infty})\\
 & = & 1-\oph(\chi)+\cO_{L^{2}\to L^{2}}(\h^{\infty})\end{eqnarray*}
since $\ti\chi$ and $1-\chi$ have disjoints supports by construction. 

We now look for the solution on $\cE^{\delta}$. From Corollary \ref{cor:Decay Propag},
we have an exponential decay of the propagator $\cU^{Nt_{0}}$ if
$N$ becomes large. To use this information, we set \[
A_{1}(z,\h,T_{\h})=\frac{\i}{\h}\int_{0}^{T_{\h}}\cU^{t}\e^{\frac{\i t}{\h}z}\oph(\chi)\, dt\]
where $T_{\h}$ has to be adjusted. Hence, \[
Q(z,\h)A_{1}(z,\h,T_{\h})=\oph(\chi)-\cU^{t}\e^{\frac{\i t}{\h}z}\vert_{t=T_{\h}}\oph(\chi)=\oph(\chi)+R_{1}\,.\]
Since $\Im z/\h=\Im\zeta=\cO(1)$, we have \begin{equation}
\|R_{1}\|=\e^{-T_{\h}\Im\zeta}\|\cU^{T_{\h}}\oph(\chi)\|.\label{eq:T1 Remainder}\end{equation}
From Corollary \ref{cor:Decay Propag}, we know that for $N=Tt_{0}\log\h^{-1}$
with $\h\leq\h_{0}$, \[
\|\cU^{Nt_{0}+1}\oph(\chi)\|\leq C\h^{-\frac{d}{2}}\e^{Nt_{0}\Pr(a^{u})^{+}}.\]
We observe that this bound is useful only if does not diverge as $\h\to0$,
which is the case if \[
T\geq\frac{d}{2t_{0}|\Pr(a^{u})^{+}|}\defeq T_{0}.\]
Let us define \begin{equation}
T_{\h}^{0}=T_{0}t_{0}\log\h^{-1}+1=\frac{d}{2|\Pr(a^{u})^{+}|}\log\h^{-1}+1\:.\label{eq:T0}\end{equation}
If $T_{\h}=T\log\h^{-1}t_{0}+1$, with $T>T_{0}$ chosen large enough,
we find \begin{eqnarray*}
\|R_{1}\| & \leq & C\h^{Tt_{0}\gamma}\h^{-\frac{d}{2}-Tt_{0}\pra}=\cO(\h^{m})\end{eqnarray*}
with $m=m(T_{\h})\geq0$, since $\gamma-\pra>0$. Consequently, there
is $T_{1}=T_{1}(\eps)>0$ such that $T_{\h}^{1}=T_{1}t_{0}\log\h^{-1}+1$
satisfies $m(T_{\h}^{1})=0.$ This means that for $T_{\h}\geq T_{\h}^{1}$,
we have \[
Q(z,\h)(A_{0}(z,\h)+A_{1}(z,\h,T_{h}))=1+\cO_{L^{2}\to L^{2}}(1),\]
in other words, $A_{0}+A_{1}$ is {}``close'' to the resolvent $Q^{-1}$
. Hence, we impose now $T_{\h}\geq T_{\h}^{1}$, and evaluate the
norms of $A_{0}$ and $A_{1}$. By construction, $\|A_{0}\|=\cO(1)$.
For $A_{1}$, we have to estimate an integral of the form \[
I_{T_{\h}}=\int_{0}^{T_{\h}}\e^{-t\Im\zeta}\|\cU^{t}\oph(\chi)\|dt.\]
Let us split the integral according to $T_{\h}^{0}$, and use the
decay of $\cU^{t}\oph(\chi)$ for $t\geq T_{\h}^{0}$ : \begin{eqnarray*}
|I_{T_{\h}}| & \leq & T_{\h}^{0}\e^{-T_{\h}^{0}\Im\zeta}+C\h^{-\frac{d}{2}}\int_{T_{\h}^{0}}^{\infty}\e^{-t\Im\zeta}\e^{(t-1)\pra}dt\\
 & \leq & T_{\h}^{0}\e^{-T_{\h}^{0}\Im\zeta}(1+C_{\eps}\h^{-\frac{d}{2}}\e^{(T_{h}^{0}-1)\pra})\\
 & \leq & C_{\eps}T_{\h}^{0}\e^{-T_{\h}^{0}\Im\zeta}\:.\end{eqnarray*}
 Using \eqref{eq:T0}, this gives \[
\|A_{1}(z,\h,T_{\h})\|\leq C_{\eps}\h^{-1+c_{0}\Im\zeta}\log\h^{-1}\]
where $C_{\eps}>0$ depends now on $M$, $a$ and $\eps$ while \begin{equation}
c_{0}=\frac{d}{2|\pra|}.\label{eq:c0}\end{equation}

\end{proof}
We now translate these results obtained in the semiclassical settings
in terms of $\tau$. Recall \[
P(\tau)=-\Delta-\tau^{2}-2\i a\tau\equiv\frac{1}{\h^{2}}Q(z,\h).\]
and set $R(\tau)\defeq P(\tau)^{-1}.$ The operator $R(\tau)$ is
directly related to the resolvent $(\tau-\cB)^{-1}$ : a straightforward
computation shows that \[
(\tau-\cB)^{-1}=\left(\begin{array}{cc}
R(\tau)(-2\i a-\tau) & -R(\tau)\\
R(\tau)(2\i a\tau-\tau^{2}) & -R(\tau)\tau\end{array}\right).\]

\begin{prop}
\label{pro:Resolvent Tau}Let $\eps>0$ be such that $\Pr(a^{u})+\eps<0$,
and $\ga<0$ satisfying \[
\Pr(a^{u})+\eps<\ga<0.\]
Let $\tau\in\C\setminus\Spec\cB$ be such that $\gamma\leq\Im\tau<0$.
Set also $\moy{\tau}=(1+|\tau|^{2})^{\fr12}$. There exists a constant
$C>0$ depending on $M,a$ and $\eps$ such that for any $\ka\geq d/2$,
we have

\begin{eqnarray*}
(i) &  & \|R(\tau)\|_{L^{2}\to L^{2}}\leq C_{\eps}\moy{\tau}^{-1-c_{0}\ga}\log\moy{\tau}\\
(ii) &  & \|R(\tau)\|_{L^{2}\to H^{2}}\leq C_{\eps}\moy{\tau}^{1-c_{0}\ga}\log\moy{\tau}\\
(iii) &  & \|R(\tau)\|_{H^{\ka}\to H^{1}}\leq C_{\eps}\\
(iv) &  & \|\tau R(\tau)\|_{H^{\ka}\to H^{0}}\leq C_{\eps}\,.\end{eqnarray*}
\end{prop}
\begin{proof}
$(i)$ follows directly from rescaling the statements of the preceding
proposition. For $(ii),$ observe that \[
\|R(\tau)u\|_{H^{2}}\leq C(\|R(\tau)u\|_{L^{2}}+\|\Delta R(\tau)u\|_{L^{2}})\,,\ C>0.\]
But \[
\|\Delta R(\tau)u\|_{L^{2}}\leq\|u\|_{L^{2}}+|\tau^{2}+2\tau\i a|\|R(\tau)u\|_{L^{2}}\,,\]
so using $(i)$, we get \[
\|R(\tau)u\|_{H^{2}}\leq C\lp(1+|\tau^{2}+2\i a\tau|)\|R(\tau)u\|_{L^{2}}+\|u\|_{L^{2}}\rp\leq C_{\eps}\moy{\tau}^{1-c_{0}\gamma}\log\moy{\tau}\|u\|_{L^{2}}\,.\]
To arrive at $(iii)$, we start from the following classical consequence
of the Hölder inequality : \begin{equation}
\|R(\tau)u\|_{H^{1-s}}^{2}\leq\|R(\tau)u\|_{H^{2}}^{1-s}\|R(\tau)u\|_{L^{2}}^{1+s}\,,\ \ s>0.\label{eq:Interpolation}\end{equation}
From $(i)$ and $(ii)$, we obtain \[
\|R(\tau)u\|_{H^{1-s}}\leq C_{\eps}\moy{\tau}^{-(\ga c_{0}+s)}\log\moy{\tau}\|u\|_{L^{2}}\,.\]
If we choose $s>-\ga c_{0}$, we get $\|R(\tau)\|_{H^{0}\to H^{1-s}}\leq C_{\eps}.$
Hence, for any $s'\geq0$ we have \[
\|R(\tau)\|_{H^{s'}\to H^{s'+1-s}}\leq C_{\eps}\,.\]
Taking $s'=s$ shows $(iii)$, where we must have $\ka>-\ga c_{0}$.
In view of \eqref{eq:c0}, and the fact that $\ga\geq\pr+\ti\eps$,
this condition is satisfied as soon as $\ka\geq d/2$. The last equation
$(iv)$ is derived as $(iii)$, by considering \[
\|\tau R(\tau)u\|_{H^{1-s}}^{2}\leq|\tau|^{2}\|R(\tau)u\|_{H^{2}}^{1-s}\|R(\tau)u\|_{L^{2}}^{1+s}\,,\ \ s>0,\]
 and choosing $s$ so that $\|\tau R(\tau)\|_{H^{0}\to H^{1-s}}^{2}\leq C_{\eps}$. 
\end{proof}

\subsection{Eigenvalues expansion }

We now prove Theorem \ref{thm:Eigenvalues expansion}. Let us fix
$\eps>0$ so that $\pr+\eps<0$. From Theorem \ref{thm:gapW} we know
that \[
\Card\lp\Spec\cB\cap(\R+\i[\pr+\eps,0])\rp\defeq n(\eps)<\infty.\]
Hence there is $e_{0}(\eps)>0$ such that $\Spec\cB\cap(\R+\i[\pr+\eps])\subset\Omega$,
where \[
\Omega=\Omega(\eps)=[-e_{0},e_{0}]+\i[\pr+\eps,0].\]
 We then call $\{\tau_{0},\dots,\tau_{n(\eps)-1}\}=\Spec\cB\cap\Omega$,
and set by convention $\tau_{0}=0$. We define as above $\pra=\pr+\eps$.
Since we look at the eigenvalues $\tau\in\Omega$, let us introduce
the spectral projectors on the generalized eigenspace $E_{j}$ for
$j\in\llb0,n-1\rrb$ :\[
\Pi_{j}=\frac{1}{2\i\pi}\oint_{\gamma_{j}}(\tau-\cB)^{-1}d\tau\,,\quad\Pi_{j}\in\cL(\cH,D(\cB^{\infty})),\]
where $\ga_{j}$ are small circles centered in $\tau_{j}$. We also
denote by \[
\Pi=\sum_{j=0}^{n}\Pi_{j}\]
the spectral projection onto $\bigoplus_{j=0}^{n}E_{j}$. We call
$E_{0}$ the eigenspace corresponding to the eigenvalue $\tau_{0}=0$.
It can be shown \cite{Leb} that $E_{0}$ is one dimensional over
$\C$ and spanned by $(1,0)$, so \[
\Pi_{0}\bs\omega=(c(\bs\omega),0)\,\ \ \ \trm{with}\ \ c(\bs\omega)\in\C.\]
 Let now $\bs\omega=(\omega_{0},\omega_{1})$ be in $\cH^{\ka}$.
Near a pole $\tau_{j}$ of $(\tau-\cB)^{-1}$, we have \[
(\tau-\cB)^{-1}=\frac{\Pi_{j}}{\tau-\tau_{j}}+\sum_{k=2}^{m_{j}}\frac{(\cB-\tau_{j})^{k-1}\Pi_{j}}{(\tau-\tau_{j})^{k}}+H_{j}(\tau)\]
where $H_{j}$ is an operator depending holomorphically on $\tau$
in a neighbourhood of $\tau_{j}$, and $m_{j}$ is the multiplicity
of $\tau_{j}$. Since $\Pi\in\cL(\cH,D(\cB^{\infty})),$ we have the
following integral representation of $\e^{-\i t\cB}\Pi\bs\omega$,
with absolute convergence in $\cH$: \begin{equation}
\e^{-\i t\cB}\Pi\bs\omega=\frac{1}{2\i\pi}\int_{-\infty+\i\alpha}^{+\infty+\i\alpha}\e^{-\i t\tau}(\tau-\cB)^{-1}\Pi\bs\omega d\tau\,,\ \ \ t>0,\ \alpha>0.\label{eq:Semigroup Integral}\end{equation}
The integrand in the right hand side has poles located at $\tau_{j}$,
$j\in\llb0,n-1\rrb$, so that \begin{eqnarray*}
\e^{-\i t\cB}\Pi\bs\omega & = & \sum_{j}\frac{1}{2\i\pi}\oint_{\gamma_{j}}\e^{-\i t\tau}\frac{\Pi_{j}}{\tau-\tau_{j}}\bs\omega d\tau\\
 & = & \sum_{j}\e^{-\i t\tau_{j}}p_{\tau_{j}}(t)\bs\omega\defeq\sum_{j}\e^{-\i t\tau_{j}}\bs u_{j}(t)\,.\end{eqnarray*}
The operators $p_{\tau_{j}}(t)$ appearing in the residues are polynomials
in $t$, with degree at most $m_{j}$, taking their values in $\cL(\cH,D(\cB^{\infty}))$.
It follows that for some $C>0$ depending only on $M$ and $a$, \[
\|\bs u_{j}(t)\|_{\cH}\leq Ct^{m_{j}}\|\bs\omega\|_{\cH}\,.\]
 The remainder term appearing in Theorem \ref{thm:Eigenvalues expansion}
is now identified : \[
\bold r_{n}(t)=\e^{-\i t\cB}(1-\Pi)\bs\omega\,.\]
To conclude the proof, we have therefore to evaluate $\|\bold r_{n}\|_{\cH}$.
To do so, we will use in a crucial way the resolvent bounds below
the real axis that we have obtained in the preceding section. We consider
the solution $u(t,x)$ of \eqref{eq:DWE} with initial data $\bs u=(u_{0},u_{1})=(1-\Pi)\bs\omega$,
with $\bs\omega\in\cH^{\ka}$, $\ka\geq d/2$. Let us define $\chi\in C^{\infty}(\R)$,
$0\leq\chi\leq1$, such that $\chi=0$ for $t\leq0$ and $\chi=1$
for $t\geq1$. If we set $v=\chi u$, we have \begin{equation}
(\d_{t}^{2}-\Delta+2a\d_{t})v=g_{1}\label{eq:u1g1}\end{equation}
where \begin{equation}
g_{1}=\chi''u+2\chi'\d_{t}u+2a\chi'u\,.\label{eq:g1}\end{equation}
Note also that $\supp g_{1}\subset[0,1]\times M$, and $v(t)=0$ for
$t\leq0$. Let us denote the inverse Fourier transform in time by
\[
\cF_{t\to-\tau}:u\mapsto\check{u}(\tau)=\int_{\R}\e^{\i t\tau}u(t)dt.\]
Applying $\cF_{t\to-\tau}$ (in the distributional sense) to both
sides of \eqref{eq:u1g1} yields to \[
P(\tau)\check{v}(\tau,x)=\check{g}_{1}(\tau,x)\,.\]
We then remark that $R(\tau)\check{g}_{1}(\tau,x)$ is the first component
of \[
\i(\tau-\cB)^{-1}\cF_{t\to-\tau}\lp\chi'(t)(u,\i\d_{t}u)\rp.\]
From the properties of $\Pi$, it is clear that the operator $(\tau-\cB)^{-1}(1-\Pi)$
depends holomorphically on $\tau$ in the half-plane $\Im\tau\geq\pra$.
From $(u,\i\d_{t}u)=\e^{-\i t\cB}(1-\Pi)\bs\omega$, we then conclude
that $\i(\tau-\cB)^{-1}\cF_{t\to-\tau}\lp\chi'(t)(u,\i\d_{t}u)\rp$
depends also holomorphically on $\tau$ in the half plane $\Im\tau\geq\pra$.
Hence $\check{v}(\tau,x)=R(\tau)\check{g}_{1}(\tau,x)$ and an application
of the Parseval formula yields to \begin{eqnarray*}
\|\e^{-t\pra}v(t,x)\|_{L^{2}(\R_{+},H^{1})} & = & \|\check{v}(\tau+\i\pra\|_{L^{2}(\R,H^{1})}\\
 & = & \|R(\tau+\i\pra)\check{g}_{1}(\tau+\i\pra,x)\|_{L^{2}(\R,H^{1})}\\
 & \leq & C_{\eps}\|\check{g}_{1}(\tau+\i\pra,x)\|_{L^{2}(\R,H^{\ka})}\\
 & \leq & C_{\eps}\|g_{1}(t,x)\|_{L^{2}(\R_{+},H^{\ka})}\,.\end{eqnarray*}
where we have used Proposition \ref{pro:Resolvent Tau}. The term
appearing in the last line can in fact be controlled by the initial
data. From \eqref{eq:g1}, we have \begin{equation}
\|g_{1}\|_{L^{2}(\R_{+};H^{\ka})}\leq C\lp\|u\|_{L^{2}([0,1];H^{\ka})}+\|\d_{t}u\|_{L^{2}([0,1];H^{\ka})}\rp\,.\label{eq:g1Bound}\end{equation}
A direct computation shows \[
\d_{t}\|u\|_{\cH^{\ka}}^{2}\leq C(\|u\|_{\cH^{\ka}}^{2}+\|\d_{t}u\|_{H^{\ka}}^{2}+\|\nabla u\|_{H^{\ka}}^{2})\,.\]
The Gronwall inequality for $t\in[0,1]$ gives \begin{eqnarray*}
\|u(t,\cdot)\|_{H^{\ka}}^{2} & \leq & C\lp\|u(0,\cdot)\|_{H^{\ka}}^{2}+\int_{0}^{t}(\|\d_{s}u(s)\|_{H^{\ka}}^{2}+\|\nabla u(s)\|_{H^{\ka}}^{2})ds\rp\\
 & \leq & C\|\bs\omega\|_{\cH^{\ka}}^{2},\end{eqnarray*}
since the $\ka-$energy \[
E^{\ka}(t,u)=\fr12(\|\d_{t}u\|_{H^{\ka}}^{2}+\|\nabla u\|_{H^{\ka}}^{2})\]
is also decreasing in $t$. Coming back to \eqref{eq:g1Bound}, we
see that $\|g_{1}\|_{L^{2}(\R_{+};H^{\ka})}\leq C\|\bs\omega\|_{\cH^{\ka}}$
and then, \[
\|\e^{-t(\pr+\eps)}v(t,x)\|_{L^{2}(\R_{+},H^{1})}\leq C_{\eps}\|\bs\omega\|_{\cH^{\ka}}.\]
This is the exponential decay we are looking for, but in the integrated
form. It is now easy to see that\[
\|u(t,\cdot)\|_{H^{1}}\leq C_{\eps}\e^{t(\pr+\eps)}\|\bs\omega\|_{\cH^{\ka}}\,.\]
We have to check that the same property is valid for $\d_{t}u$. Using
the same methods as above, we also have \[
P(\tau)\cF_{t\to-\tau}(\d_{t}v)=-\tau\check{g}_{1}(\tau),\]
and then, $\cF_{t\to-\tau}(\d_{t}v)=-\tau R(\tau)\check{g}_{1}(\tau)$.
It follows that\begin{eqnarray*}
\|\e^{-t\pra}\d_{t}v(t,x)\|_{L^{2}(\R_{+},H^{0})} & = & \|\check{v}(\tau+\i\pra\|_{L^{2}(\R,H^{0})}\\
 & = & \|\tau R(\tau+\i\pra)\check{g}_{1}(\tau+\i\pra,x)\|_{L^{2}(\R,H^{0})}\\
 & \leq & C_{\eps}\|\check{g}_{1}(\tau+\i\pra,x)\|_{L^{2}(\R,H^{\ka})}\\
 & \leq & C_{\eps}\|g_{1}(t,x)\|_{L^{2}(\R_{+},H^{\ka})}\,.\end{eqnarray*}
Grouping the results, we see that \[
\|\bold u\|_{\cH}\leq C_{\eps}\e^{t(\pr+2\eps)}\|\bs\omega\|_{\cH^{\ka}}\;\]
and this concludes the proof of Theorem \ref{thm:Eigenvalues expansion}.

\subsection{Energy decay}

We end this section with the proof of the Theorem \ref{thm:energy},
which gives the exponential energy decay. This is an immediate consequence
of the following lemma, that tells us that the energy can be controlled
by the $H^{1}$ norm of $u$, for $t\geq2$:
\begin{lem}
\label{lem:Energy}There exists $C>0$ such that for any solution
$u$ of \eqref{eq:DWE} and $E(u,t)$ the associated energy functional,
we have \[
E(u,T)\leq C\|u\|_{L^{2}([T-2,T+1];H^{1})}^{2}\,,\ \ T\geq2.\]
\end{lem}
\begin{proof}
This is a standard result, we borrow the proof from \cite{EZ}. For
$T>2$, we choose $\chi_{2}\in C^{\infty}(\R)$, $0\leq\chi_{2}\leq1$
such that $\chi_{2}(t)=1$ for $t\geq T$ and $\chi_{2}(t)=0$ if
$t\leq T-1$. Setting $u_{2}(t,x)=\chi_{2}(t)u(t,x)$, we have \[
(\d_{t}^{2}-\Delta+2a\d_{t})u_{2}=g_{2}\]
 for $g_{2}=\chi_{2}''u+2\chi'_{2}\d_{t}u+2a\chi'_{2}u$. Note that
$g_{2}$ is compactly supported in $t$. Define now \[
E_{2}(u,t)=\fr12\int_{M}(|\d_{t}u_{2}|^{2}+|\nabla u_{2}|^{2})\dvol\]
and compute

\begin{eqnarray*}
E'_{2}(u,t) & = & \scal{\d_{t}^{2}u_{2}}{\d_{t}u_{2}}-\scal{\Delta u_{2}}{\d_{t}u_{2}}\\
 & = & -2\scal{a\d_{t}u_{2}}{\d_{t}u_{2}}+\scal{g_{2}}{\d_{t}u_{2}}\\
 & \leq & C\int_{M}|\d_{t}u_{2}|(|\d_{t}u|+|u|)\dvol\\
 & \leq & C\lp E_{2}(u,t)+\int_{M}(|\d_{t}u|^{2}+|u|^{2})\dvol\rp\,.\end{eqnarray*}
We remark that $E_{2}(u,T-1)=0$ and $E_{2}(u,T)=E(u,T)$, so the
Gronwall inequality on the interval $[T-1,T]$ gives \begin{equation}
E(u,T)\leq C\lp\|\d_{t}u\|_{L^{2}([T-1,T];L^{2})}^{2}+\|u\|_{L^{2}([T-1,T];L^{2})}^{2}\rp\,.\label{eq:EnergyL2}\end{equation}
To complete the proof, we need to bound the term $\|\d_{t}u\|_{L^{2}([T-1,T];L^{2})}^{2}$.
For this purpose, we choose $\chi_{3}\in C^{\infty}(\R)$, $0\leq\chi_{3}\leq1$
such that $\chi_{3}(t)=1$ for $t\in[T-1,T]$ and $\chi_{3}(t)=0$
if $t\leq T-2$ and $t\geq T+1$. From \eqref{eq:DWE}, we get 

\begin{eqnarray*}
0 & = & \int_{T-2}^{T+1}\langle\chi_{3}^{2}u,\d_{t}^{2}u-\Delta u+2a\d_{t}u\rangle dt\\
 & = & \int_{T-2}^{T+1}-\chi_{3}^{2}\langle\d_{t}u,\d_{t}u\rangle-2\chi_{3}\chi_{3}'\langle u,\d_{t}u\rangle+2\chi_{3}^{2}\langle u,a\d_{t}u\rangle+\chi_{3}^{2}\langle u,-\Delta u\rangle dt\,,\end{eqnarray*}
whence\[
\|\d_{t}u\|_{L^{2}([T-1,T];L^{2})}\leq C\|u\|_{L^{2}([T-2,T+1];H^{1})}.\]
Substituting this bound in \eqref{eq:EnergyL2} yields to the result.
\end{proof}
The Theorem \ref{thm:energy} follows now from the preceding lemma.
Let us denote by $u_{j}(t,x)$ and and $r_{n}(t,x)$ the first component
of $p_{\tau_{j}}(t)\bs\omega$ and $\e^{-\i t\cB}(1-\Pi)\bs\omega$
respectively. We learned above that \[
u(t,x)=\sum_{j=0}^{n}\e^{-\i t\tau_{j}}u_{j}(t,x)+r_{n}(t,x)\]
with $\|u_{j}(t,\cdot)\|_{H^{1}}\leq Ct^{m_{j}}\|\bs\omega\|_{\cH^{\ka}}$,
and $\|r_{n}(t,\cdot)\|_{H^{1}}\leq C_{\eps}\e^{t(\pr+2\eps)}\|\bs\omega\|_{\cH^{\ka}}$.
Suppose first that the projection $ $of $\bs\omega$ on $E_{0}$
vanishes, i.e. $\Pi_{0}\bs\omega=0$. Then, from the preceding lemma
we clearly have \[
E(u,t)^{\fr12}\leq\sum_{j=1}^{n}\e^{t\Im\tau_{j}}C\|u_{j}(t,x)\|_{H^{1}}+C_{\eps}\e^{t(\pr+2\eps)}\|\bs\omega\|_{\cH^{\ka}}\,.\]
This shows Theorem \ref{thm:energy} when $\Pi_{0}\bs\omega=0$. But
the general case follows easily: we can write $\ti u(t,x)=u(t,x)-\Pi_{0}\bs\omega$
for which we have the expected exponential decay, and notice that
$E(\ti u,t)=E(u,t)$ since $\Pi_{0}\bs\omega$ is constant.

\section{Hyperbolic dispersion estimate\label{sec:HypDispEst}}

This last section is devoted to the proof of Proposition \ref{pro:HypDispEst}.
Let $\eps,\delta,\cV,\cW$ and $\oph(\chi)$ be as in Section \ref{sec:Spectral Gap}.
We also set $N=T\log\h^{-1}$, $T>0$.

\subsection{Decomposition into elementary Lagrangian states\label{sub:Decomposition}}

Recall that each set $\cW_{\beta}\equiv\cW_{b_{0}\dots b_{t_{0}-1}}$
in the cover $\cW$ has the property\begin{equation}
\Phi^{k}(\cW_{\beta})\subset\cV_{b_{k}}\,,\ k\in\llb0,t_{0}-1\rrb\label{eq:visits}\end{equation}
for some sequence $b_{0},b_{1},\dots,b_{t_{0}-1}$. We will say that
a sequence sequence $\beta_{1}\dots\beta_{N}$ of sets $\cW_{\beta_{k}}$
is adapted to the dynamics if the following condition is satisfied
:\[
\forall k\in[1,N-1],\ \ \Phi^{kt_{0}}(\cW_{\beta_{1}})\cap\cW_{\beta_{k+1}}\neq\emptyset.\]
In this case, we can associate to the sequence $\{\beta_{i}\}$ a
sequence $\gamma_{1},\dots,\gamma_{Nt_{0}}$ of sets $\cV_{\gamma_{k}}\subset\cV$
which are visited at the times $0,\dots,Nt_{0}-1$ for some points
of $\cW_{\beta_{1}}$. We will only consider the sequences adapted
to the dynamics. Indeed, it is clear from standard results on propagation
of singularities that \[
\|\cU_{\beta_{N}}\dots\cU_{\beta_{1}}\|=\cO(\h^{\infty})\]
if the sequence is not adapted (see Appendix \ref{sec:semiclass}),
and in this case, Proposition \ref{pro:HypDispEst} is obviously true.

We now decompose further each evolution of length $t_{0}$ in \eqref{eq:DecompPropag}
by inserting additional quantum projectors. To unify the notations,
we define for $j\in\llb1,Nt_{0}\rrb$ the following projectors and
the corresponding open sets in $T^{*}M$ : \begin{equation}
\fkP_{\gamma_{j}}=\begin{cases}
\Pi_{\beta_{k}} & \trm{if}\ j-1=kt_{0},\ k\in\N\\
\ti\Pi_{\gamma_{j}} & \trm{if}\ j-1\neq0\mod t_{0}\end{cases}\,,\ \ \fkV_{\gamma_{j}}=\begin{cases}
\cW_{\beta_{k}} & \trm{if}\ j-1=kt_{0},\ k\in\N\\
\cV_{\gamma_{j}} & \trm{if}\ j-1\neq0\mod t_{0}\,.\end{cases}\label{eq:Ouverts Generalises}\end{equation}
We will also denote by $\mathsf{F}_{\gamma}\in C_{0}^{\infty}(T^{*}M)$
the function such that $\supp\FF_{\gamma}\subset\fkV_{\gamma}$ and
$\fkP_{\gamma}=\mathsf{F}_{\gamma}^{w}$.

Let us set up also a notation concerning the constants appearing in
the various estimates we will deal with. Let $\ell,K\in\N$ be two
parameters (independent of $\h$), and $e_{1},e_{2},e_{3}>0$ some
fixed numbers. For a constant $C$ depending on $M$ and derivatives
of $\chi,$ $a$, $\Phi^{t}$ (for $t$ bounded) up to order $e_{1}\ell+e_{2}K+e_{3}$,
we will write $C^{(\ell,K)}(M,\chi)$, or simply $C^{(K)}(M,\chi)$
if only one parameter is involved. If the constant $C$ depends also
on the cutoff functions $\FF_{\gamma}$ and their derivatives, we
will write \[
C=C^{(\ell,K)}(M,\chi,\cV).\]
This is to recall us the dependence on the cutoff function $\chi$
supported inside $\cE^{\delta}$, and the refined cover $\cV$. We
will sometimes use the notation $C^{(\ell,K)}(M,\cV)$ when no dependence
on $\chi$ is assumed. Note that $\cV$ depends implicitely on $\eps$
since its diameter was chosen such that \eqref{eq:PressApprox} holds.

Using \eqref{eq:visits}, standard propagation estimates give \[
\cU^{t_{0}}\Pi_{\beta_{1}}=\cU\fkP_{\gamma_{t_{0}}}\dots\cU\fkP_{\ga_{1}}+\cO_{L^{2}\to L^{2}}(\h^{\infty})\,,\ \ \ \cU\equiv\cU^{1}\,,\]
and similar properties for $\cU^{t_{0}}\Pi_{\beta_{k}},\ k>1$. Finally,
\begin{equation}
\cU_{\be_{Nt_{0}}}\dots\cU_{\beta_{1}}\cU\oph(\chi)=\cU\fkP_{\gamma_{Nt_{0}}}\dots\cU\fkP_{\gamma_{1}}\cU\oph(\chi)+\cO_{L^{2}\to L^{2}}(\h^{\infty})\,.\label{eq:Proj-Evol}\end{equation}
Take now $\Psi\in L^{2}(M)$. In order to show Proposition \ref{pro:HypDispEst},
we will write $\oph(\chi)\Psi$ as a linear decomposition over some
elementary Lagrangian states, and study the individual evolution of
such elementary states by $\cU^{Nt_{0}+1}$. This type of method comes
back to \cite{Ana1} and is the key tool to prove Proposition \ref{pro:HypDispEst}.
The decomposition of $\oph(\chi)\Psi$ is obtained by expliciting
the action of $\oph(\chi)$ in local coordinates (see Appendix \ref{sec:semiclass}).
When applying $\oph(\chi)$ to $\Psi$ using local charts labelled
by $\ell$, we get \begin{eqnarray*}
[\Op_{\h}(\chi)\Psi](x) & = & \sum_{\ell}\frac{1}{(2\pi\h)^{d}}\int\e^{\i\frac{\langle\eta,x-z_{0}\rangle}{\h}}\chi(\frac{x+z_{0}}{2},\eta)\varphi_{\ell}(z_{0})\phi_{\ell}(x)\Psi(z_{0})d\eta\: dz_{0}\\
 & = & \sum_{\ell}\int\delta_{\chi,z_{0}}^{\ell}(x)\Psi(z_{0})dz_{0}\,,\end{eqnarray*}
where we have defined \[
\delta_{\chi,z_{0}}^{\ell}(x)\defeq\frac{1}{(2\pi\h)^{d}}\int\e^{\i\frac{\langle\eta,x-z_{0}\rangle}{\h}}\chi(\frac{x+z_{0}}{2},\eta)\varphi_{\ell}(z_{0})\phi_{\ell}(x)d\eta\,.\]
This is a Lagrangian state, which Lagrangian manifold is given by
\[
\Lambda^{0}\defeq T_{z_{0}}^{*}M\cap\cE^{\delta}\subset T^{*}M\,.\]
Geometrically, $\Lambda^{0}$ corresponds to a small, connected piece
taken out of the union of spheres $\{T_{z_{0}}^{*}M\cap p^{-1}(\fr12+\nu)\:,\ |\nu|\leq\delta\}$.
If we project and evolve $\Psi$ according to the operator appearing
in the right hand side of \eqref{eq:Proj-Evol}, we get :

\begin{eqnarray}
\|\cU^{t}\fkP_{\gamma_{Nt_{0}}}\dots\cU\fkP_{\gamma_{1}}\cU\oph(\chi)\Psi\| & \leq & \sum_{\ell}\sup_{z}\|\cU^{t}\fkP_{\gamma_{Nt_{0}}}\dots\fkP_{\gamma_{1}}\cU\delta_{\chi,z_{0}}^{\ell}\|\int_{M}|\Psi(x)|dx\nonumber \\
 & \leq & C\sum_{\ell}\sup_{z}\|\cU^{t}\fkP_{\gamma_{Nt_{0}}}\dots\fkP_{\gamma_{1}}\cU\delta_{\chi,z_{0}}^{\ell}\|\|\Psi\|\label{eq:superposition}\end{eqnarray}
where $C>0$ depends only on the manifold $M$. Hence we are lead
by this superposition principle to study in detail states of the form
$\cU^{t}\fkP_{\gamma_{n}}\dots\cU\fkP_{\gamma_{1}}\cU\delta_{\chi,z_{0}}^{\ell}$,
for $n\in\llb1,Nt_{0}\rrb$ and $t\in[0,1]$. For simplicity, because
the local charts will not play any role in the following, we will
omit them in the formul\ae.

\subsection{Evolution of Lagrangian states and their Lagrangian manifolds\label{sub:Time1}}

\subsubsection{Ansatz for short times}

In this section we investigate the first step of the sequence of projection--evolution
given in \eqref{eq:Proj-Evol}: our goal is to describe the state
$\cU^{t}\delta_{\chi,z_{0}}$ with $t\in[0,1]$. Since $\cU^{t}$
is a Fourier integral operator, we know that $\cU^{t}\delta_{\chi,z_{0}}$
is a Lagrangian state, supported on the Lagrangian manifold \[
\Lambda^{0}(t)\defeq\Phi^{t}(\Lambda^{0})\,,\ t\in[0,1].\]
Because of our assumptions on the injectivity radius, the flow $\Phi^{t}:\Lambda^{0}(s)\to\Lambda^{0}(t)$
for $1\geq t\geq s>0$, induces on $M$ a bijection from $\pi\Lambda^{0}(s)$
to $\pi\Lambda^{0}(t)$. In other words, $\Lambda^{0}(t)$ projects
diffeomorphically on $M$ for $t\in]0,1]$, i.e. $\ker d\pi|_{\Lambda^{0}(t)}=0$
: in this case, we will say that $\Lambda^{0}(t)$ is \textit{projectible}.
This is the reason for introducing a first step of propagation during
a time $1$ : the Lagrangian manifold $\Lambda^{0}(0)$ is not projectible,
but as soon as $t\in]0,1]$, $\Lambda^{0}(t)$ projects diffeomorphically.
Treating separately this evolution for times $t\in[0,1]$ avoid some
unnecessary technical complications. 

The remark above implies that the Lagrangian manifold $\Lambda^{0}(t),\ t\in]0,1]$
is generated by the graph of the differential of a smooth, well defined
function $S_{0}$ : \[
\Lambda^{0}(t)=\{(x,d_{x}S_{0}(t,x,z_{0}))\,:\ 1\geq t>0,\ x\in\pi\Phi^{t}(\Lambda^{0})\}\,.\]
This means that for $t\in]0,1]$, we have the Lagrangian Ansatz :\begin{eqnarray}
v^{0}(t,x,z_{0}) & \defeq & \cU^{t}\delta_{\chi,z_{0}}(x)\nonumber \\
 & = & \frac{1}{(2\pi\h)^{\frac{d}{2}}}\lp\e^{\i\frac{S_{0}(t,x,z_{0})}{\h}}\sum_{k=0}^{K-1}\h^{k}b_{k}^{0}(t,x,z_{0})+\h^{K}B_{K}^{0}(t,x,z_{0})\rp.\label{eq:Ansatz1}\end{eqnarray}
The functions $b_{k}^{0}(t,x,z_{0})$ are smooth, and $x\in\pi\Lambda^{0}(t)$.
Furthermore, given any multi index $\ell$, they satisfy \begin{equation}
\|\d_{x}^{\ell}b_{k}^{0}(t,\cdot,z_{0})\|\leq C_{\ell,k}\label{eq:Ck Norms : time 0}\end{equation}
 where the constants $C_{\ell,k}$ depends only on $M$ (via the Hamiltonian
flow of $p$), the damping $a$, the cutoff function $\chi$ and their
derivatives up to order $2k+\ell$. However, note that $C_{0,0}$
only depends on $M$. The remainder satisfies $\|B_{K}^{0}\|\leq C_{K}$
where the constant $C_{K}$ also depends on $M$, $a$, $\chi$ and
is uniformly bounded with respect to $x,z_{0}$. The base point $z_{0}$
will be fixed until section \ref{sub:main}, so it will be ommited
in the following to simplify the notations.

\subsubsection{Further evolution}

In the sequence of projection--evolution \eqref{eq:Proj-Evol}, we
then have performed the first step, and obtained an Ansatz for $\cU^{t}\delta_{\chi},\ t\in]0,1]$
up to terms of order $\h^{K-d/2}$, for any $K\geq0$. The main goal
of the next paragraphs consist in finding an Ansatz for the full state
$ $ \begin{equation}
v^{n}(t,x)\defeq\cU^{t}\fkP_{\ga_{n}}\cU\fkP_{\ga_{n-1}}\dots\cU\fkP_{\ga_{1}}\cU\delta_{\chi}\,,\ t\in[0,1],\ n\geq1\,.\label{eq:vn}\end{equation}
The $\beta_{j}$ are defined according to $j-1\mod t_{0}$ as in the
preceding section, but here $n$ is arbitrary in the interval $\llb1,Nt_{0}\rrb$.
Because the operator $\cU^{t}\fkP$ is a Fourier integral operator,
$v^{j}(t,x)$, $j\geq1$ is a Lagrangian state, with a Lagrangian
manifold which will be denoted by $\Lambda^{j}(t)$. This manifold
consist in a small piece of $\Phi^{j+t}(\Lambda^{0})$, because of
the successive applications of the projectors $\fkP_{\gamma}$ between
the evolution operator $\cU$. If $j=1$, the Lagrangian manifold
$\Lambda^{1}(0)$ is given by \[
\Lambda^{1}(0)=\Lambda^{0}(1)\cap\fkV_{\ga_{1}},\]
and for $t\in[0,1]$ we have $\Lambda^{1}(t)=\Phi^{t}(\Lambda^{1}(0))$.
For $j\geq1$, $\Lambda^{j}(t)$ can be obtained by a similar procedure:
knowing $\Lambda^{j-1}(1)$, we take for $\Lambda^{j}(t),\ t\in[0,1]$
the Lagrangian manifold \[
\Lambda^{j}(0)\defeq\Lambda^{j-1}(1)\cap\fkV_{\gamma_{j}}\,,\ \ \trm{and}\ \ \Lambda^{j}(t)=\Phi^{t}(\Lambda^{j}(0))\,.\]
Of course, if the intersection $\Lambda^{j-1}(1)\cap\fkV_{\gamma_{j}}$
is empty, the construction has to be stopped, since by standard propagation
estimates, $v^{j}$ will be of order $\cO(\h^{\infty})$. But this
situation will not happen since the sequence $\{\beta_{i}\}$ is adapted
to the dynamics. It follows that \[
\forall j\in\llb1,n\rrb,\quad\Lambda^{j}(0)\neq\emptyset\,.\]
One can show (see \cite{AN}, Section 3.4.1 for an argument) that
the Lagrangian manifolds $\Lambda^{j}(t)$ are projectible for all
$j\geq1$. This is mainly because $M$ has no conjugate points. In
particular, any $\Lambda^{j}(t)$ can be parametrized as a graph on
$M$ of a differential, which means that there is a generating function
$S_{j}(t,x)$ such that \[
\Lambda^{j}(t)=\{x,d_{x}S_{j}(t,x)\}\,.\]
By extension, we will call a Lagrangian state projectible if its Lagrangian
manifold is. 

Let us introduce now some notations that will be often used later.
Suppose that $x\in\pi\Lambda^{j}(t)$, $j\geq1$. Then, there is a
unique $y=y(x)\in\pi\Lambda^{j}(0)$ such that \[
\pi\circ\Phi^{t}(y,d_{y}S_{j}(0,y))=x\,.\]
If we denote for $t\in[0,s]$ the (inverse) induced flow on $M$ by
\[
\phi_{S_{j}(s)}^{-t}:x\in\pi\Lambda^{j}(s)\mapsto\pi\Phi^{-t}\lp x,d_{x}S_{j}(s,x)\rp\in\pi\Lambda^{j}(s-t),\]
we have $y(x)=\phi_{S_{j}(t)}^{-t}(x).$ If $x\in\pi\Lambda^{j}(t)$,
then by construction \[
\Phi^{-t-k}(x,d_{x}S_{j}(t,x))\in\Lambda^{j-k}(0)\subset\Lambda^{j-k-1}(1)\,,\ k\in\llbracket0,j-1\rrbracket\,.\]
By definition, we will write \[
\phi_{S_{j}(t)}^{-t-k}(x)=\pi\Phi^{-t-k}(x,d_{x}S_{j}(t,x))\quad\et\quad\phi_{S_{j}}^{-k}(x)=\pi\Phi^{-k}(x,d_{x}S_{j}(1,x))\,.\]
To summarize, our sequence of projections and evolutions can be cast
into the following way:\begin{equation}
\xymatrix{\delta_{\chi}\ar[r]^{\cU^{1}} & v^{0}(1,\cdot)\ar[r]^{\fkP_{1}} & v^{1}(0,\cdot)\ar[r]^{\cU} & v^{1}(1,\cdot)\ar[r]^{\fkP_{2}} & \dots\ar[r]^{\fkP_{n}} & v^{n}(0,\cdot)\ar[r]^{\cU^{t}} & v^{n}(t,\cdot)\\
\Lambda^{0}\ar[r]^{\Phi^{1}} & \Lambda^{0}(1)\ar[r]^{|_{\fkV_{1}}} & \Lambda^{1}(0)\ar[r]^{\Phi^{1}} & \Lambda^{1}(1)\ar[r]^{|_{\fkV_{2}}} & \dots\ar[r]^{|_{\fkV_{n}}} & \Lambda^{n}(0)\ar[r]^{\Phi^{t}} & \Lambda^{n}(t)}
\label{eq:Diagram}\end{equation}
On the top line are written the successive evolutions of the Lagrangian
states, while the evolution of their respective Lagrangian manifolds
is written below (the notation $|_{\fkV}$ denotes a restriction to
the set $\fkV\subset T^{*}M$).

\subsection{Evolution of a projectible Lagrangian state}

Let $\fkV_{\gamma}$ and $\fkP_{\ga}$ be as in \eqref{eq:Ouverts Generalises}.
The next proposition contains an explicit description of the action
of the Fourier integral operators $\cU^{t}\fkP$ on projectible Lagrangian
states localized inside $\msV_{\gamma}$. 
\begin{prop}
\label{pro:BasicEvol}Let $\VV{}$ and $\PP{}=\FF_{\gamma}^{w}$ be
as in Section \ref{sub:Decomposition}. Let $w_{\h}(x)=w(x)\e^{\frac{\i}{\h}\psi(x)}$
be a projectible Lagrangian state, supported on a projectible Lagrangian
manifold \[
\Lambda=\{x,d_{x}\psi(x)\}\subset\fkV_{\gamma}\,.\]
Assume also that $\Lambda(t)\defeq\Phi^{t}\Lambda$ is projectible
for $t\in[0,1]$. We have the following asymptotic developement :\begin{equation}
[\cU^{t}\PP{}w_{\h}](x)=\e^{\frac{\i}{\h}\psi(t,x)}\sum_{k=0}^{K-1}\h^{k}w_{k}(t,x)+\h^{K}r_{K}(t,x)\label{eq:DPS}\end{equation}
where $\psi(t,\cdot)$ is a generating function for $\Lambda(t)$.
The amplitudes $w_{k}$ can be computed from the geodesic flow (via
the function $\varphi_{\ga}$), the damping $q$ and the function
$\FF_{\gamma}$. Moreover, the following bounds hold :\begin{eqnarray*}
\|w_{k}\|_{C^{\ell}} & \leq & C_{\ell,k}\|w\|_{C^{\ell+2k}}\\
\|r_{K}\|_{C^{\ell}} & \leq & C_{\ell,K}\|w\|_{C^{\ell+2K+d}}\end{eqnarray*}
where the constants depend on $\varphi_{\ga}$, $a$, $\FF_{\gamma}$
and their derivatives up to order $\ell+2K+d$, namely $C_{\ell,k}=C^{(\ell,k)}(M,\cV)$.
An explicit expression for $w_{k}$ will be given in the proof.\end{prop}
\begin{proof}
The steps we will encounter below are very standard in the non-damping
case, i.e. $q=0$. If the diameter of the partition $\cV$ of $\cE^{\delta}$
is chosen small enough, we can assume without loss of generality the
existence of a function $\varphi_{\gamma}\in C^{\infty}([0,1]\times\R^{d}\times\R^{d})$
which generates the canonical transformation given by the geodesic
flow on $\msV_{\gamma}$ for times $t\in[0,1]$, in other words :
\begin{equation}
\forall(y,\eta)\in\fkV_{\gamma}\,,\ \ \Phi^{t}(y,\eta)=(x,\xi)\Leftrightarrow\xi=\d_{x}\varphi_{\gamma}(t,x,\eta)\ \trm{and}\ y=\d_{\eta}\varphi_{\gamma}(t,x,\eta)\,.\label{eq:Generating Function}\end{equation}
Furthermore, $\varphi_{\gamma}$ satisfies $\det\d_{x,\eta}^{2}\varphi_{\gamma}\neq0$,
and solves the following Hamilton-Jacobi equation \[
\begin{cases}
\d_{t}\varphi_{\gamma}+p(x,d_{x}\varphi_{\ga})=0\\
\varphi_{\ga}(0,x,\eta)=\langle\eta,x\rangle\,.\end{cases}\]
We first look for an oscillatory integral representation: \begin{eqnarray}
\cU^{t}\mathsf{P}_{\gamma}w_{\h}(x) & = & \frac{1}{(2\pi\h)^{d}}\iint\e^{\frac{\i}{\h}(\varphi_{\gamma}(t,x,\eta)-\scal y{\eta}+\psi(y))}\sum_{k=0}^{K-1}\h^{k}a_{k}^{\ga}(t,x,y,\eta)w(y)dyd\eta\label{eq:OscIntg}\\
 &  & +\cO_{L^{2}}(\h^{K})\nonumber \\
 & \defeq & b_{\h}(t,x)+\h^{K}\ti r_{K}(t,x)\,,\ \ \ \|\ti r_{K}\|=\cO(1),\nonumber \end{eqnarray}
with $(y,\eta)\in\fkV_{\ga}$ . For simplicity, we will omit the dependence
on $\gamma$ in the formul\ae. We have to determine the amplitudes
$a_{k}$. For this, we want $b_{\h}$ to solve \[
\frac{\d b_{\h}}{\d t}=(\frac{\i\h\Delta_{g}}{2}-q)b_{\h}\]
 up to order $\h^{K}$. Direct computations using \eqref{eq:laplacian}
show that the functions $\varphi$ and $a_{k}$ must satisfy the following
equations :\begin{equation}
\begin{cases}
\d_{t}\varphi+p(x,d_{x}\varphi)=0 & \trm{(Hamilton-Jacobi\ equation)}\\
\d_{t}a_{0}+qa_{0}+X[a_{0}]+\fr12a_{0}\div_{g}X=0 & \trm{(0-th\ transport\ equation)}\\
\d_{t}a_{k}+qa_{k}+X[a_{k}]+\fr12a_{k}\div_{g}X=\frac{\i}{2}\Delta_{g}a_{k-1} & \trm (k-\trm{th\ transport\ equation)}\end{cases}\label{eq:wkb}\end{equation}
with initial conditions \[
\begin{cases}
\begin{array}{c}
\varphi(0,x,\eta)=\scal x{\eta}\\
a_{0}(0,x,y,\eta)=\FF(\frac{x+y}{2},\eta)\\
a_{k}(0,x,y,\eta)=0\ \trm{for\ }k\geq1.\end{array}\end{cases}\]
The variables $y$ and $\eta$ are fixed in these equations, so they
will play the role of parameters for the moment and will sometimes
be skipped in the formul\ae. $X$ is a vector field on $M$ depending
on $t$, and $\div_{g}X$ its Riemannian divergence. In local coordinates,
\[
X=g^{ij}(x)\d_{x_{j}}\varphi(t,x)\:\d_{x_{i}}=\d_{\xi_{i}}p(x,\d_{x}\varphi(t,x))\d_{x_{i}}\quad\et\quad\div_{g}X=\frac{1}{\sqrt{\bar{g}}}\d_{i}(\sqrt{\bar{g}}X^{i}).\]
The Hamilton-Jacobi equation is satisfied by construction. To deal
with the transport equations, we notice that $X$ corresponds to the
projection on $M$ of the Hamiltonian vector field $H_{p}$ at $(x,d_{x}\varphi(t,x,\eta))\in T^{*}M$.
Let us call first \[
\Lambda_{t,\eta}=\{(x,d_{x}\varphi(t,x,\eta)),x\in\pi\Phi^{t}\Lambda\},\ \ \eta\ \trm{fixed}.\]
 This Lagrangian manifold is the image of the Lagrangian manifold
$\Lambda_{0,\eta}=\{(y,\eta):y\in\pi\Lambda\}$ by the geodesic flow
$\Phi^{t}$. The flow $\ka_{s}^{t}$ on $M$ generated by $X$ can
be now identified with the geodesic flow restricted to $\Lambda_{s,\eta}$:
\[
\ka_{s}^{t}:\pi\Lambda_{s,\eta}\owns x\mapsto\pi\Phi^{t}(x,\d_{x}\varphi(t,x,\eta))\in\pi\Lambda_{t+s,\eta}.\]
The inverse flow $(\ka_{s}^{t})^{-1}$ will be denoted by $\ka_{s+t}^{-t}$.
Let us extend now the flow $\ka_{s}^{t}$ of $X$ on $M$ to the flow
$\cK^{t}$ generated by the vector field $\cX=\d_{t}+X$ on $\R\times M$
: \[
\cK^{t}:\begin{cases}
\R\times M\to\R\times M\\
(s,x)\mapsto(s+t,\ka_{s}^{t}(x))\,.\end{cases}\]
We then identify the functions $a_{k}$ with Riemannian half-densities
on $\R\times M$ -- see \cite{Dui,EZ}: \[
a_{k}(t,x)\equiv a_{k}(t,x)\sqrt{dt\dvol(x)}=a_{k}(t,x)\sqrt{\bar{g}(x)}^{\fr12}|dtdx|^{\fr12}\in C^{\infty}(\R\times M,\Omega_{\fr12}).\]
Since we have \[
\cL_{\cX}(a_{k}\sqrt{dt\dvol})=(\cX[a_{k}]+\fr12a_{k}\div_{g}X)\sqrt{dt\dvol}\,,\]
the 0--th transport equation takes the simple form of an ordinary
differential equation: \[
\cL_{\cX}(a_{0}\sqrt{dt\dvol})+qa_{0}\sqrt{dt\dvol}=0.\]
This is the same as \[
\frac{d}{dt}(\cK^{t})^{*}a_{0}\sqrt{dt\dvol}=-qa_{0}\sqrt{dt\dvol},\]
which is solved by\[
a_{0}\sqrt{dt\dvol}=\e^{-\int_{0}^{t}q\circ\cK^{s-t}ds}(\cK^{-t})^{*}a_{0}\sqrt{dt\dvol}\,.\]
We now have to make  explicit the coordinates dependence, which yields
to\[
a_{0}(t,x)\sqrt{\bar{g}(x)}^{\fr12}|dxdt|^{\fr12}=\e^{-\int_{0}^{t}q\circ\ka_{t}^{s-t}(x)ds}a_{0}(0,\ka_{t}^{-t}(x))\sqrt{\bar{g}(\ka_{t}^{-t}(x))}|\det d_{x}\ka_{t}^{-t}|^{\fr12}|dxdt|^{\fr12}.\]
Consequently, \foreignlanguage{french}{\[
a_{0}(t,x)=\e^{-\int_{0}^{t}q\circ\ka_{t}^{s-t}(x)ds}a_{0}(0,\ka_{t}^{-t}(x))\frac{\sqrt{\bar{g}(\ka_{t}^{-t}(x))}}{\sqrt{\bar{g}(x)}}|\det d_{x}\ka_{t}^{-t}|^{\fr12}\,.\]
}Since \[
\ka_{t}^{-t}:x\mapsto\pi\Phi^{-t}(x,\d_{x}\varphi(t,x,\eta))=\d_{\eta}\varphi(t,x,\eta),\]
it is clear that $|\det d_{x}\ka_{t}^{-t}(x)|=|\det\d_{x\eta}^{2}\varphi(t,x,\eta)|.$
For convenience, we introduce the following operator $\cT_{s}^{t}$
transporting functions $f$ on $M$ with support inside $\pi\Lambda_{s,\eta}$
to functions on $\pi\Lambda_{t+s,\eta}$ while damping them along
the trajectory :\[
\cT_{s}^{t}(f)(x)=\e^{-\int_{0}^{t}q\circ\ka_{t+s}^{\sigma-t}d\sigma}f(\ka_{t+s}^{-t}(x))\frac{\sqrt{\bar{g}(\ka_{t+s}^{-t}(x))}}{\sqrt{\bar{g}(x)}}|\det d_{x}\ka_{t+s}^{-t}(x)|^{\fr12}.\]
This operator plays a crucial role, since we have\begin{equation}
a_{0}(t,\cdot)=\cT_{0}^{t}(a_{0}(0,\cdot))=\cT_{0}^{t}\FF,\label{eq:a0}\end{equation}
from which we see that $a_{0}(t,\cdot)$ is supported inside $\pi\Lambda_{t,\eta}$.
By the Duhamel formula, the higher order terms can now be computed,
they are given by\[
a_{k}(t,\cdot)=\int_{0}^{t}\cT_{s}^{t-s}\lp\frac{\i}{2}\Delta_{g}a_{k-1}(s)\rp\: ds\,.\]
The ansatz $b_{\h}(t,x)$ constructed so far satisfies the approximate
equation\[
\frac{\d b_{\h}}{\d t}=(\i\h\Delta_{g}-q)b_{\h}-\frac{\i}{2}\h^{K}\iint\e^{\frac{\i}{\h}S(t,x,\eta,y)}w(y)\Delta_{g}a_{K-1}(t,x,y,\eta)\: dyd\eta\,.\]
The difference with the actual solution $\cU^{t}\fkP$ is bounded
by \[
\h^{K}t\|\Delta_{g}a_{K-1}\|\leq Ct\h^{K}\,,\]
 where $C=C^{(2K)}(M,\cV)$, so \eqref{eq:OscIntg} is satisfied. 

As noticed above, for time $t>0$, the state $\cU^{t}\fkP w_{\h}$
is a Lagrangian state, supported on the Lagrangian manifold $\Lambda(t)=\Phi^{t}\Lambda$.
By hypothese, $\Lambda(t)$ is projectible, so we expect an asymptotic
expansion for $b_{\h}(t,x)$, exactly as in \eqref{eq:Ansatz1}. To
this end, we now proceed to the stationary phase developement of the
oscillatory integral in \eqref{eq:OscIntg}. We set \[
I_{k}(x)=\frac{1}{(2\pi\h)^{d}}\iint\e^{\frac{\i}{\h}(\varphi(t,x,\eta)-\langle y,\eta\rangle+\psi(y))}a_{k}(t,x,y,\eta)w(y)dyd\eta\,.\]
The stationary points of the phase are given by \[
\begin{cases}
\psi'(y)=\eta\\
\d_{\eta}\varphi(t,x,\eta)=y,\end{cases}\]
for which there exists a solution $(y_{c},\eta_{c})\in\Lambda(0)$
in view of \eqref{eq:Generating Function}. Moreover, this solution
is unique since $\Lambda(t)$ is projectible: $y_{c}=y_{c}(x)\in\pi\Lambda(0)$
is the unique point in $\pi\Lambda(0)$ such that $x=\pi\Phi^{t}(y_{c},\psi'(y_{c}))\,,$
and then $\eta_{c}=\psi'(y_{c})$ is the unique vector allowing the
point $y_{c}$ to reach $x$ in time $t$. The generating function
for $\Lambda(t)$ we are looking for is then given by \[
\psi(t,x)=S(t,x,y_{c,}(x),\eta_{c}(x)).\]
Applying now the stationary phase theorem for each $I_{k}$ (see for
instance \cite{Hor}, Theorem 7.7.6, or \cite{NZ}, Lemma 4.1 for
a similar computation), summing up the results and ordering the different
terms according to their associated power of $\h$, we see that \eqref{eq:DPS w}
holds with\[
w_{0}(t,x)=\e^{\frac{\i}{\h}\beta(t)}\frac{a_{0}(t,x,y_{c},\eta_{c})}{|\det(1-\d_{\eta\eta}^{2}\varphi(t,x,\eta_{c})\circ\psi''(y_{c}))|^{\fr12}}w(y_{c})\,,\ \ \ \ \beta\in C^{\infty}(\R),\]
and\begin{equation}
w_{k}(t,x)=\sum_{i=0}^{k}A_{2i}(x,D_{x,\eta})(a_{k-i}(t,x,y,\eta)w(y))|_{(y,\eta)=(y_{c},\eta_{c})}\,.\label{eq:DPS w}\end{equation}
$A_{2i}$ denotes a differential operator of order $2i$, with coefficients
depending smoothly on $\varphi$, $\psi$ and their derivatives up
to order $2i+2$. This yields to the following bounds :\begin{eqnarray*}
\|w_{k}\|_{C^{\ell}} & \leq & C_{\ell,k}\|w\|_{C^{\ell+2k}}\end{eqnarray*}
where $C_{\ell,k}=C^{(\ell,k)}(M,\cV)$. The remainder terms $r_{K}(t,x)$
is the sum of the remainders coming from the stationary phase developement
of $I_{k}$ up to order $K-k$. Each remainder of order $K-k$ has
a $C^{\ell}$ norm bounded by $C_{\ell,K-k}\h^{K-k}\|w\|_{C^{\ell+2(K-k)+d}}$
so we see that \[
\|r_{K}\|_{C^{\ell}}\leq C_{\ell,K}\|w\|_{C^{\ell+2K+d}},\ \ C=C^{(\ell,K)}(M,\cV).\]
The principal symbol $w_{0}$ can also be interpreted more geometrically.
As in Section \ref{sub:Time1}, denote by $\phi_{\psi(t)}^{-t}$ the
following map \[
\phi_{\psi(t)}^{-t}:\begin{cases}
\pi\Lambda(t)\to\pi\Lambda(0)\\
x\mapsto\pi\Phi^{-t}(x,d_{x}\psi(t,x))\,.\end{cases}\]
Let us write the differential of $\Phi^{t}:(y,\eta)\mapsto(x,\xi)$
as $d\Phi^{t}(\delta y,\delta\eta)=(\delta x,\delta\xi).$ Using \eqref{eq:Generating Function},
we have \begin{eqnarray*}
\delta y & = & \d_{x\eta}^{2}\varphi\delta x+\d_{\eta\eta}^{2}\varphi\delta\eta\\
\delta\xi & = & \d_{xx}^{2}\varphi\delta x+\d_{x\eta}^{2}\varphi\delta\eta\,,\end{eqnarray*}
and then, since $\d_{x\eta}^{2}\varphi$ is invertible,\foreignlanguage{french}{\[
\left(\begin{array}{c}
\delta x\\
\delta\xi\end{array}\right)=\left(\begin{array}{cc}
\d_{x\eta}^{2}\varphi^{-1} & -\d_{x\eta}^{2}\varphi^{-1}\d_{\eta\eta}^{2}\varphi\\
\d_{xx}^{2}\varphi\d_{x\eta}^{2}\varphi^{-1} & \d_{x\eta}^{2}\varphi-\d_{xx}\varphi\d_{x\eta}^{2}\varphi^{-1}\d_{\eta\eta}^{2}\varphi\end{array}\right)\left(\begin{array}{c}
\delta y\\
\delta\eta\end{array}\right)\]
} If we restrict $\Phi^{t}$ to $\Lambda(0)$, we have $\delta\eta=\psi''(y)\delta y$,
which means that for $x\in\pi\Lambda(t)$, \[
d\phi_{\psi(t)}^{-t}(x)=\d_{x\eta}^{2}\varphi(t,x,\eta_{c})(1-\d_{\eta\eta}^{2}\varphi(t,x,\eta_{c})\psi''(y_{c}))^{-1}.\]
It follows from \eqref{eq:DPS w} that \begin{eqnarray*}
w_{0}(t,x) & = & \e^{\frac{\i}{\h}\beta(t)}w(y_{c})\FF(y_{c},\eta_{c})\:\e^{-\int_{0}^{1}q(\phi_{\psi(t)}^{-t+s}(x))ds}|\det d\phi_{\psi(t)}^{-t}(x)|^{\fr12}\sqrt{\frac{\bar{g}(\phi_{\psi(t)}^{-t}(x))}{\bar{g}(x)}}^{\fr12}\\
 & = & \e^{\frac{\i}{\h}\beta(t)}w(y_{c})\mathsf{F}(y_{c},\eta_{c})\:\e^{-\int_{0}^{1}q(\phi_{\psi(t)}^{-t+s}(x))ds}\lva\Jac(d\phi_{\psi(t)}^{-t}(x))\rva^{\fr12}\,,\end{eqnarray*}
where $\Jac(f)$ denotes the Jacobian of $f:M\to M$ measured with
respect to the Riemannian volume.
\end{proof}

\subsection{Ansatz for $n>1$}

In this paragraph, we construct by induction on $n$ a Lagrangian
state $b^{n}(t,x)$ supported on $\Lambda^{n}(t)$, in order to approximate
$v^{n}(t,x)$ up to order $\h^{K-d/2}$. 
\begin{prop}
\label{pro:Ansatze}There exists a sequence of functions \[
\{b_{k}^{n}(t,x),S_{n}(t,x)\ :\ n\geq1,\ k<K,\ x\in M,\ t\in[0,1]\}\]
such that $S_{n}(t,x)$ is a generating function for $\Lambda^{n}(t)$,
and \begin{equation}
v^{n}(t,x)=\frac{1}{(2\pi\h)^{\frac{d}{2}}}\e^{\i\frac{S_{n}(t,x)}{\h}}\sum_{k=0}^{K-1}\h^{k}b_{k}^{n}(t,x)+\h^{K-\frac{d}{2}}R_{K}^{n}(t,x)\label{eq:AnsatzV}\end{equation}
where $R_{K}^{n}$ satisfies \begin{equation}
\|R_{K}^{n}\|\leq C_{K}(1+C\h)^{n}\lp\sum_{i=2}^{n}\sum_{k=0}^{K-1}\|b_{k}^{i-1}(1,\cdot)\|_{C^{2(K-k)+d}}+C'\rp\label{eq:Remainder Brut}\end{equation}
where $C'=C^{(K)}(M,\chi)$, $C_{K}=C^{(K)}(M,\chi,\cV)$ and $C>0$
is fixed. \end{prop}
\begin{proof}
The construction of the amplitudes $b_{k}^{n}$ for all $k\geq0$
is done by induction on $n$, following step by step the sequence
\eqref{eq:Diagram}. In Section \ref{sub:Time1} we obtained $\cU^{1}\delta_{\chi}$
as a projectible Lagrangian state: \begin{eqnarray*}
v^{0}(1,x) & = & \frac{1}{(2\pi\h)^{d/2}}\e^{\i\frac{S_{0}(1,x)}{\h}}\sum_{k=0}^{K-1}\h^{k}b_{k}^{0}(1,x)+\h^{K-d/2}B_{K}^{0}(1,x)\\
 & \defeq & \frac{1}{(2\pi\h)^{d/2}}b^{0}(1,x)+\h^{K-d/2}R_{K}^{0}(1,x)\,,\end{eqnarray*}
and we know that $b^{0}(1,\cdot)$ satisfies the hypotheses of Proposition
\ref{pro:BasicEvol}, which will be used to describe $\cU^{t}\fkP_{\ga_{1}}v^{0}(1,\cdot)$.
More generally, suppose that the preceding step has lead for some
$n\geq1$ to \begin{eqnarray*}
v^{n-1}(t,x) & = & \frac{1}{(2\pi\h)^{\frac{d}{2}}}\e^{\i\frac{S_{n-1}(t,x)}{\h}}\sum_{k=0}^{K-1}\h^{k}b_{k}^{n-1}(t,x)+\h^{K-d/2}R_{K}^{n-1}(t,x)\\
 & = & \frac{1}{(2\pi\h)^{\frac{d}{2}}}b^{n-1}(t,x)+\h^{K-d/2}R_{K}^{n-1}(t,x)\end{eqnarray*}
 where $b^{n-1}(t,\cdot)$ is a Lagrangian state, supported on the
Lagrangian manifold $\Lambda^{n-1}(t)$, and $R_{K}^{n-1}$ is some
remainder in $L^{2}(M)$. We now apply Proposition \ref{pro:BasicEvol}
to each Lagrangian state $\e^{\frac{\i}{\h}S_{n-1}(1,x)}\h^{k}b_{k}^{n-1}(1,x)$
appearing in the definition of $b^{n-1}$. Because of the term $\h^{k}$,
if we want an Ansatz as in \eqref{eq:AnsatzV}, it is enough to describe
$\cU^{t}\fkP_{\gamma_{n}}v_{k}^{n-1}(1,\cdot)$ up to order $K-k$,
which gives a remainder of order $C_{K-k}\h^{K-k}\|b_{k}^{n-1}(1,\cdot)\|_{C^{2(K-k)+d}}$.
Grouping the terms corresponding to the same power of $\h$ when applying
Proposition \ref{pro:BasicEvol} to each $(v_{k}^{n-1})_{0\leq k<K-1}$
yields to \[
[\cU^{t}\PP nb^{n-1}](x)=\e^{\frac{\i}{\h}S_{n}(t,x)}\sum_{k=0}^{K-1}\h^{k}b_{k}^{n}(t,x)+\h^{K}B_{K}^{n}(t,x)\defeq b^{n}(t,x)+\h^{K}B_{K}^{n}(t,x),\]
where $S_{n}(t,x)$ is a generating function of the Lagrangian manifold
\[
\Lambda^{n}(t)=\Phi^{t}(\Lambda^{n-1}(1)\cap\fkV_{\gamma_{n}}).\]
 The coefficients $b_{k}^{n}$ are given by \begin{equation}
b_{k}^{n}(t,x)=\sum_{i=0}^{k}\sum_{l=0}^{k-i}A_{2l}(a_{k-i-l}^{\ga_{n}}(t,x,y,\eta)b_{i}^{n-1}(1,y))|_{(y,\eta)=(y_{c},\eta_{c})}\label{eq:bnk}\end{equation}
where $y_{c}=\phi_{S_{n}(t)}^{-t}(x)\,,\ \eta_{c}=d_{y}S_{n-1}(1,y_{c})$.
In particular, $b_{0}^{n}(t,x)=\cD_{n}(t,x)b_{0}^{n-1}(1,y_{c})$,
with\begin{equation}
\cD_{n}(t,x)=\e^{-\int_{0}^{1}q(\phi_{S_{n}(t)}^{s-t}(x))ds}\lva\Jac(d\phi_{S_{n}(t)}^{-t}(x))\rva^{\fr12}\e^{\frac{\i}{\h}\beta_{n}(t)}\mathsf{F}_{\gamma_{n}}(y_{c},\eta_{c})\label{eq:Symbole Principal}\end{equation}
 for some $\beta_{n}(t)\in C^{\infty}(\R)$. The remainder $B_{K}^{n}$
satisfies \begin{equation}
\|B_{K}^{n}(t)\|\leq C_{K}\sum_{k=0}^{K-1}\|b_{k}^{n-1}(1,\cdot)\|_{C^{2(K-k)+d}}\:,\ \ C_{K}=C^{(K)}(M,\chi,\cV).\label{eq:RnN}\end{equation}
Hence, we end up with \begin{eqnarray*}
v^{n}(1,x) & = & \frac{1}{(2\pi\h)^{\frac{d}{2}}}\lp\e^{\i\frac{S_{n}(1,x)}{\h}}\sum_{k=0}^{K-1}\h^{k}b_{k}^{n}(1,x)+\h^{K}(B_{K}^{n}(1,x)+\cU^{1}\fkP_{\ga_{1}}B_{K}^{n-1}(1))\rp\\
 & \defeq & \frac{1}{(2\pi\h)^{\frac{d}{2}}}b^{n}(1,x)+\h^{K-d/2}R_{K}^{n}(1,x)\end{eqnarray*}
where $R_{K}^{n}(1,x)=(2\pi)^{-\frac{d}{2}}(B_{K}^{n}(1,x)+\cU^{1}\fkP_{\ga_{n}}R_{K}^{n-1}(1,\cdot))$.
Again, $b^{n}$ satisfies the hypotheses of Proposition \ref{pro:BasicEvol},
so we can continue iteratively. To complete the proof, we now have
to take all the remainders into account. From the discussion above,
we get :\begin{eqnarray*}
\cU^{t}\PP nv^{n-1}(1,\cdot) & = & (2\pi\h)^{-\frac{d}{2}}\cU^{t}\PP nb^{n-1}(1,\cdot)+\h^{K-d/2}\cU^{t}\PP n(R_{K}^{n-1}(1,\cdot))\\
 & = & (2\pi\h)^{-\frac{d}{2}}\lp b^{n}(t,\cdot)+\h^{K}B_{K}^{n}(t.\cdot)\rp+\h^{K-d/2}\cU^{t}\PP n(R_{K}^{n-1}(1,\cdot))\\
 & = & (2\pi\h)^{-\frac{d}{2}}b^{n}(t,\cdot)+\h^{K-d/2}R_{K}^{n}(t,\cdot)\end{eqnarray*}
where we defined $R_{K}^{n}=(2\pi)^{-\frac{d}{2}}B_{K}^{n}+\cU^{t}\PP n(R_{K}^{n-1})$.
Since $|\mathsf{F}_{\ga_{n}}|\leq1$, we have \[
\|\cU^{t}\fkP_{\ga_{n}}\|_{L^{2}\to L^{2}}\leq1+C\h,\ \ C>0.\]
This implies that $\|\cU^{t}\PP n(R_{K}^{n-1}(1,\cdot))\|\leq(1+C\h)\|R_{K}^{n-1}(1,\cdot)\|$,
and finally $R_{K}^{n}$ satisfies \begin{equation}
\|R_{K}^{n}\|\leq(1+C\h)^{n}\lp\|B_{K}^{n}\|+\|B_{K}^{n-1}\|+\dots\|B_{K}^{1}\|+\|B_{K}^{0}\|\rp\label{eq:Remainders}\end{equation}
In view of \eqref{eq:RnN} and \eqref{eq:Ansatz1}, this concludes
the proof. 
\end{proof}
Given $v^{n-1}(1,\cdot)$, we have then constructed $v^{n}(t,x)$
as in \eqref{eq:AnsatzV}, but it remains to control the remainder
$R_{K}^{n}$ in $L^{2}$ norm : from \eqref{eq:RnN} and \eqref{eq:Remainders},
we see that it is crucial for this to estimate properly the $C^{\ell}$
norms of the coefficients $b_{k}^{j}$ for $j\geq1$ and $k\in\llbracket0,K-1\rrbracket$.
\begin{lem}
\label{lem:Cnorms}Let $n\geq1$, and define \[
\bs\cD_{n}=\sup_{x\in\pi\Lambda^{n}(1)}\lva\prod_{i=0}^{n-1}\cD_{n-i}(1,\phi_{S_{n}}^{-i}(x))\rva\,,\quad\bs\cD_{0}=1.\]
If $x\in\pi\Lambda^{n}(1)$, the principal symbol $b_{0}^{n}$ is
given by \begin{equation}
b_{0}^{n}(1,x)=\lp\prod_{j=0}^{n-1}\cD_{n-j}(1,\phi_{S_{n}}^{-j}(x))\rp\, b_{0}^{0}(1,\phi_{S_{n}}^{-n}(x)).\label{eq: Ppal Symbol}\end{equation}
For $k\in\llbracket0,K-1\rrbracket$, the functions $b_{k}^{n}$ satisfy 

\begin{equation}
\|b_{k}^{n}(1,\cdot)\|_{C^{\ell}}\leq C_{k,\ell}(n+1)^{3k+\ell}\bs\cD_{n}\label{eq:CNorms}\end{equation}
where $C_{k,\ell}=C^{(\ell,k)}(M,\chi,\cV)$. It follows that \begin{eqnarray}
\|B_{K}^{n}(1,\cdot)\| & \leq & C_{K}n^{3K+d}\bs\cD_{n-1}\label{eq:DecayProd1}\\
\|R_{K}^{n}(1)\| & \leq & C_{K}(1+C\h)^{n}\sum_{j=1}^{n}j^{3K+d}\bs\cD_{j-1}\label{eq:DecayProd2}\end{eqnarray}
where $C>0$ and $C_{K}=C^{(K)}(M,\chi,\cV)$. On the other hand,
if $x\notin\pi\Lambda^{n}(1),$ we have $b_{k}^{n}(x)=0$ for $k\in\llbracket0,K-1\rrbracket.$ \end{lem}
\begin{proof}
First, if $x\notin\pi\Lambda^{n}(1)$, then there is no $\rho\in\fkV_{\gamma_{n}}$
such that $\pi\Phi^{1}(\rho)=x$, and then $v^{n}(1,x)=\cO(\h^{\infty})$.
In what follows, we then consider the case $x\in\pi\Lambda^{n}(1)$.
We first see that \eqref{eq: Ppal Symbol} simply follows from \eqref{eq:Symbole Principal}
applied recursively. If $\rho_{n}=(x_{n},\xi_{n})\in\Lambda^{n}(1)$,
we call $\rho_{j}=(x_{j},\xi_{j})=\Phi^{j-n}(\rho_{n})\in\Lambda^{j}(1)$
if $j\geq0$. In other words,\[
\forall j\in\llbracket1,n\rrbracket,\ \ x_{j-1}=\phi_{S_{j}}^{-1}(x_{j}).\]
It will be useful to keep in mind the following sequence, which illustrates
the backward trajectory of $\rho_{n}\in\Lambda^{n}(1)$ under $\Phi^{-k}$,
$k\in\llb1,n\rrb$ and its projection on $M$ : 

$$
\xymatrix{
\rho_0\in\Lambda^{0}(1)\ar[d]_\pi&
\ar[l]_-{\Phi^{-1}}\ar[d]_\pi\rho_1\in\Lambda^{1}(1)&
\ar[l]_-{\Phi^{-1}}\dots&
\ar[l]_-{\Phi^{-1}}\ar[d]_\pi\rho_{n-1}\in\Lambda^{n-1}(1)&
\ar[l]_-{\Phi^{-1}}\ar[d]_\pi \rho_n\in\Lambda^{n}(1)\\
x_0 &
\ar[l]_-{\phi^{-1}_{S_1}}x_1 &
\ar[l]_-{\phi^{-1}_{S_2}}\dots &
\ar[l]_-{\phi^{-1}_{S_{n-1}}} x_{n-1}&
\ar[l]_-{\phi^{-1}_{S_n}}x_n
 }
$$We denote schematically the Jacobian matrix $d\phi_{S_{j}}^{-i}=\frac{\d x_{j-i}}{\d x_{j}}$
for $1\leq i\leq j\leq n$. Since for any $E>0$, the sphere bundle
$T_{z}^{*}M\cap p^{-1}(E)$ is transverse to the stable direction
\cite{Kli}, the Lagrangians $\Lambda^{n}\subset\Phi^{n}\Lambda^{0}$
converge exponentially fast to the weak unstable foliation as $n\to\infty$.
This implies that $\Phi^{t}|_{\Lambda^{0}}$ is asymptotically expanding
as $t\to\infty$, except in the flow direction. Hence, the inverse
flow $\Phi^{-t}|_{\Lambda^{n}}$ acting on $\Lambda^{n}$ and its
projection $\phi_{S_{n}}^{-t}$ on $M$ have a tangent map uniformly
bounded with respect to $n,t$. As a result, the Jacobian matrices
$\d x_{j-i}/\d x_{j}$ are uniformly bounded from above : for $1\leq i\leq j\leq n$
there exists $C=C(M)$ independent of $n$ such that \begin{equation}
\lno\frac{\d x_{j-i}}{\d x_{j}}\rno\leq C\,.\label{eq:Jacob Bounds}\end{equation}
It follows that if we denote $D_{j}=\sup_{x_{j}}\cD_{j}(1,x_{j})$,
there exists $C=C(M)>0$ such that \begin{equation}
C^{-1}\leq D_{j}\leq C.\label{eq:Uniform Dn}\end{equation}
Note also that \[
\sup_{x\in\pi\Lambda^{n}(1)}\lva\prod_{j=0}^{n-1}\cD_{n-j}(1,\phi_{S_{n}}^{-j}(x))\rva=\prod_{j=0}^{n-1}D_{n-j}=\bs\cD_{n}.\]
We first establish the following crucial estimate :
\begin{lem}
\label{lem:Flow Derivatives}Let $n\geq1$, and $k\in\llb1,n\rrb$.
For every multi index $\alpha$ of length $|\alpha|\geq2$, there
exists a constant $C_{\alpha}>0$ depending on $M$ such that \begin{equation}
\lno\frac{\d^{\alpha}x_{n-k}}{\d x_{n}^{\alpha}}\rno\leq C_{\alpha}k^{\alpha-1}\label{eq:Flow Derivatives}\end{equation}

\begin{proof}
We proceed by induction on $k$, from $k=1$ to $k=n$. The case $k=1$
is clear. Let us assume now that \[
\lno\frac{\d^{\alpha}x_{n-k'}}{\d x_{n}^{\alpha}}\rno\leq C_{\alpha}k'^{\alpha-1}\,,\quad k'\in\llb1,k-1\rrb\]
and let us show the bound for $k'=k$. For simplicity, we will denote
\[
\d_{j}^{\alpha}\defeq\frac{\d^{\alpha}}{\d x_{j}^{\alpha}},\ \ \d^{\alpha}x_{j}=\frac{\d^{\alpha}x_{j}}{\d x_{j+1}^{\alpha}}\,.\]
In particular, $\|\d^{\alpha}x_{j}\|\leq C_{\alpha}.$ We also recall
the Faà di Bruno formula : let $\Pi$ be the set of partitions of
the ensemble $\{1,...,|\alpha|\}$, and for $\pi\in\Pi$, write $\pi=\{B_{1},...B_{k}\}$
where $B_{i}$ is some subset of $\{1,...,|\alpha|\}$. Here $|\alpha|\geq k\geq1$,
and we denote $|\pi|=k$. For two smooth functions $g:\R^{d}\mapsto\R^{d}$
and $f:\R^{d}\mapsto\R^{d}$ such that $f\circ g$ is well defined,
one has \begin{equation}
\d^{\alpha}f\circ g\equiv\sum_{\pi\in\Pi}\d^{|\pi|}f(g)\prod_{B\in\pi}\d^{B}g\,.\label{e:fdb}\end{equation}
The term in the right hand side is written schematically, to indicates
a sum of derivatives of $f$ of order $|\pi|$, times a product of
$|\pi|$ terms, each of them corresponding to derivatives of $g$
of order $|B|$. It is important for our purpose to note that $\sum|B|=|\alpha|$.
Continuing from theses remarks, we compute \begin{eqnarray*}
X_{k}\defeq\d_{n}^{\alpha}x_{n-k} & = & \d x_{n-k}\d_{n}^{\alpha}x_{n-k+1}\\
 &  & +\sum_{\pi\in\Pi,|\pi|>1}\d^{|\pi|}x_{n-k}\prod_{B\in\pi}\d_{n}^{B}x_{n-k+1}\defeq\d x_{n-k}X_{k-1}+Y_{k-1}\,.\end{eqnarray*}
By the induction hypothesis, \begin{equation}
\|Y_{i}\|\leq C_{\alpha}i{}^{\alpha-2}\label{eq:Y bounds}\end{equation}
 since the partitions $\pi$ involved in the sum contains at least
two elements. Setting $M_{k-1}=\d x_{n-k}$ , we have\begin{eqnarray*}
X_{k} & = & M_{k-1}\dots M_{1}X_{1}+M_{k-2}\dots M_{1}Y_{1}+M_{k-3}\dots M_{1}Y_{2}\\
 &  & +\dots+M_{1}Y_{k-1}\,.\end{eqnarray*}
 From the chain rule we have \[
\frac{\d x_{j-i}}{\d x_{j}}=\frac{\d x_{j-i}}{\d x_{j-i+1}}\dots\frac{\d x_{j-1}}{\d x_{j}},\]
 and \eqref{eq:Jacob Bounds} yields to $\|M_{i-1}\dots M_{1}\|=\cO(1)$
for $2\leq i\leq k$. Adding up all the terms contributing to $X_{k}$
and taking \eqref{eq:Y bounds} into account yields to \[
\|X_{k}\|\leq C_{\alpha}(1+1^{\alpha-2}+2^{\alpha-2}+\dots+(k-1)^{\alpha-2})\leq C_{\alpha}k^{\alpha-1}\]
and the lemma is proved. 
\end{proof}
\end{lem}
We now prove \eqref{eq:CNorms}. For this, we will proceed in two
steps. First, we show the bounds for the principal symbol $b_{0}^{n}$.
Then, we treat the higher order terms $b_{k}^{n},k\geq1$ using the
bounds on $\|b_{0}^{n}\|_{C^{\ell}}$ for any $\ell$. For $b_{0}^{n}$,
The $C^{0}$ norm estimate follows directly from \eqref{eq: Ppal Symbol}.
From now on, we denote for convenience \[
\cD_{0}(x_{0})\defeq b_{0}^{0}(1,x_{0}).\]
Computing \[
\d_{n}^{\ell}b_{0}^{n}(x_{n})=\d_{n}^{\ell}(\cD_{n}(x_{n})\dots\cD_{1}(x_{1})\cD_{0}(x_{0}))\,,\]
we will obtain a sum of terms, each of them of the form\[
M_{\alpha_{n}\dots\alpha_{0}}=\d_{n}^{\alpha_{n}}\cD_{n}\d_{n}^{\alpha_{n-1}}\cD_{n-1}\dots\d_{n}^{\alpha_{1}}\cD_{1}\d_{n}^{\alpha_{0}}\cD_{0},\]
 with $\alpha_{n}+\dots+\alpha_{0}=\ell$. Note that if $\ell$ is
fixed with respect to $n$, most of the multi-indices $\alpha_{i}$
vanish when $n$ becomes large : actually, at most $|\ell|$ are non-zero,
and we will denote them by $\alpha_{i_{1}},\dots,\alpha_{i_{k}}$,
$k\leq|\ell|$. Hence the above expression is made of long strings
of $\cD_{i}$ , alternating with some derivative terms $\d_{n}^{\alpha_{i}}\cD_{i}$
which number depends only on $\ell$. We can then write \begin{equation}
\|M_{\alpha_{n}\dots\alpha_{0}}\|_{C^{0}}\leq\bs\cD_{n}\times\frac{\|\d_{n}^{\alpha_{i_{1}}}\cD_{i_{1}}\dots\d_{n}^{\alpha_{i_{k-1}}}\cD_{i_{k-1}}\d_{n}^{\alpha_{i_{k}}}\cD_{i_{k}}\|_{C^{0}}}{D_{i_{1}}\dots D_{i_{k}}}\,.\label{eq:Strings 1}\end{equation}
Let us examinate each terms $\d_{n}^{\alpha}\cD_{i}$ appearing in
the right hand side individually. By the Faà di Bruno formula and
Lemma \ref{lem:Flow Derivatives}, we have for $i\neq0$ \begin{equation}
\d_{n}^{\alpha}\cD_{i}(x_{i})=\sum_{\pi}\d_{i}^{|\pi|}\cD_{i}\prod_{B\in\pi}\d_{n}^{B}x_{i}\leq C_{\pi}n^{\alpha-|\pi|}\leq C_{\alpha}n^{\alpha-1}\label{eq:Strings 2}\end{equation}
where $C_{\alpha}=C^{(\ell,K)}(M,\chi,\cV)$. Of course, if $i=0$,
$\|\d_{0}^{\alpha}\cD_{0}(x_{0})\|_{C^{0}}\leq C_{\alpha}\|\d_{0}^{\alpha}b_{0}^{0}\|_{C^{0}}$
for some constant $C_{\alpha}>0$. Now, for a fixed configuration
of derivatives $\{\alpha\}=\{\alpha_{i_{1}},\dots\alpha_{i_{k}}\}$
we have to choose $i_{1},\dots,i_{k}$ indices among $n+1$ to form
the right hand side in \eqref{eq:Strings 1}, and the number of such
choices is at most of order $\cO((n+1)^{k})$. Hence, \begin{eqnarray}
\|\d_{n}^{\ell}b_{0}^{n}\|_{C^{0}} & \leq & \bs\cD_{n}\:\sum_{\{\alpha\}}\sum_{i_{i},\dots,i_{k}}\frac{\|\d_{n}^{\alpha_{i_{1}}}\cD_{i_{1}}\dots\d_{n}^{\alpha_{i_{k-1}}}\cD_{i_{k-1}}\d_{n}^{\alpha_{i_{k}}}\cD_{i_{k}}\|_{C^{0}}}{D_{i_{1}}\dots D_{i_{k}}}\nonumber \\
 & \leq & \bs\cD_{n}\sum_{\{\alpha\}}C_{\alpha}(n+1)^{k}(n+1)^{\alpha_{1}-1}\dots(n+1)^{\alpha_{k}-1}\nonumber \\
 & \leq & C_{\ell}\bs\cD_{n}(n+1)^{\ell}\label{eq: Cl norm -b0}\end{eqnarray}
where $C_{\ell}=C^{(\ell)}(M,\chi,\cV)$. For higher order terms $(b_{k}^{n},\ k>0)$,
we remark from \eqref{eq:bnk} that we can write \begin{equation}
b_{k}^{n}(x_{n})=\cD_{n}(x_{n})b_{k}^{n-1}(x_{n-1})+\sum_{j=1}^{k}\sum_{|\alpha|\leq2j}\Gamma_{j\alpha}^{n}(x_{n})\d_{n-1}^{\alpha}b_{k-j}^{n-1}(x_{n-1})\,.\label{eq:bnk - developped}\end{equation}
The function $\Gamma_{j\alpha}^{n}$ can be expressed with the flow,
the damping and the cutoff function $\FF_{\gamma_{n}}$. It follows
that the norms $\|\Gamma_{j\alpha}^{n}\|_{C^{\ell}}$ are uniformly
bounded with respect to $n$: \[
\|\Gamma_{j,\alpha}^{n}\|_{C^{\ell}}=C^{(\ell,K)}(M,\chi,\cV).\]
In order to show the bounds \eqref{eq:CNorms} for $k>0$, we will
proceed by induction on the index $k$. The case $k=0$ has been treated
above. Suppose now that for any $\ell$ and $k'\in\llb0,k-1\rrb$
we have proven \[
\|\d_{n}^{\ell}b_{k'}^{n}\|_{C^{0}}\leq C_{\ell}(n+1)^{3k'+\ell}\bs\cD_{n},\quad C_{\ell}=C^{(\ell,K)}(M,\chi,\cV)\,.\]
As above, to treat the case $k'=k$, we begin by the situation where
$\ell=0$. To shorten the formul\ae, we introduce for $1\leq i\leq j\leq n$
the functions \begin{eqnarray*}
\bs\Gamma^{i,k}(x_{n}) & = & \sum_{j=1}^{k}\sum_{|\alpha|\leq2j}\Gamma_{j\alpha}^{i}(x_{i})\d_{i-1}^{\alpha}b_{k-j}^{i-1}(x_{i-1})\\
\bs J_{i}^{j}(x_{n}) & = & \cD_{j}(x_{j})\cD_{j-1}(x_{j-1})\dots\cD_{i}(x_{i})\end{eqnarray*}
where the $x_{i}$, $i\leq n$ have to be considered as functions
of $x_{n}$, namely $x_{i}=\phi_{S_{n}}^{-n+i}(x_{n})$. Iterating
\eqref{eq:bnk - developped} further, we have :\begin{eqnarray}
b_{k}^{n}(x_{n}) & = & \bs J_{n}^{n}b_{k}^{n-1}(x_{n-1})+\bs\Gamma^{n,k}\nonumber \\
 & = & \bs J_{n}^{n}(\bs J_{n-1}^{n-1}b_{k}^{n-2}(x_{n-2})+\bs\Gamma^{n-1,k})+\bs\Gamma^{n,k}\nonumber \\
 & = & \bs J_{n-1}^{n}b_{k}^{n-2}(x_{n-1})+\bs J_{n}^{n}\bs\Gamma^{n-1,k}+\bs\Gamma^{n,k}\nonumber \\
 & = & \bs J_{1}^{n}b_{k}^{0}(x_{0})+\bs J_{2}^{n}\bs\Gamma^{1,k}+\bs J_{3}^{n}\bs\Gamma^{2,k}+\dots+\bs J_{n}^{n}\bs\Gamma^{n-1,k}+\bs\Gamma^{n,k}\label{eq:bnk - Sum}\end{eqnarray}
By the induction hypothesis and \eqref{eq:Uniform Dn}, each term
$\bs\Gamma^{i,k}$, $i>0$ satisfies \[
\|\bs\Gamma^{n-i,k}\|_{C^{0}}\leq C_{k}(n-i)^{3k-1}\bs\cD_{n-i}\]
hence adding up all the terms we get \[
\|b_{k}^{n}\|_{C^{0}}\leq C_{k}\bs\cD_{n}(b_{k}^{0}(x_{0})+\sum_{i=0}^{n-1}(n-i)^{3k-1})\leq C_{k}\bs\cD_{n}(n+1)^{3k}\]
and we obtain the bounds \eqref{eq:CNorms} for $\ell=0$. To evaluate
$\d^{\ell}b_{k}^{n}$, $\ell>1$, we start from the expression \eqref{eq:bnk - Sum}.
We notice first that \[
\d_{n}^{\beta}\bs\Gamma^{n-i,k}=\sum_{\beta_{1}+\beta_{2}=\beta}\sum_{j=1}^{k}\sum_{|\alpha|\leq2j}(\d_{n}^{\beta_{1}}\Gamma_{j\alpha}^{n-i}(x_{n-i}))(\d_{n}^{\beta_{2}}\d^{\alpha}b_{k-j}^{n-i-1}(x_{n-i-1}))\,.\]
Using the Faà di Bruno formula and Lemma \ref{lem:Flow Derivatives},
we get \[
\|\d_{n}^{\beta_{1}}\Gamma_{j\alpha}^{n-i,k}(x_{n-i})\|_{C^{0}}\leq C_{\beta_{1}}i^{\beta_{1}-1}\quad\et\quad\|\d_{n}^{\beta_{2}}\d^{\alpha}b_{k-j}^{n-i-1}(x_{n-i-1})\|_{C^{0}}\leq C_{\beta_{2}}i{}^{3k-1+\beta_{2}},\]
and this implies \[
\|\d_{n}^{\beta}\bs\Gamma^{n-i,k}\|_{C^{0}}\leq C_{\beta}i^{3k-1+\beta}.\]
Then, exactly the same strategy used to derive \eqref{eq: Cl norm -b0}
shows that \[
\|\d_{n}^{\ell}\bs J_{i+1}^{n}\bs\Gamma^{i,k}\|_{C^{0}}\leq C_{\ell}n^{3k-1+\ell}.\]
Using these estimates and \eqref{eq:bnk - Sum} yields to\[
\|\d_{n}^{\ell}b_{k}^{n}\|_{C^{0}}\leq C_{\ell}(n+1)n^{3k-1+\ell}\leq C_{\ell}(n+1)^{3k+\ell},\]
where the constant $C_{\ell}$ is such that $C_{\ell}=C^{(\ell,K)}(M,\chi,\cV)$.

\end{proof}

\subsection{The main estimate : proof of Proposition \ref{pro:HypDispEst}\label{sub:main}}

As noted before, the Lagrangians $\Lambda^{n}$ converge exponentially
fast as $n\to\infty$ to the weak unstable foliation. This implies
that for $x\in\pi\Lambda^{j}(1)$, the Jacobians $J_{S_{j}}(x)\defeq|\det\phi_{S_{j}(1)}^{-1}(x)|$
satisfy \[
\forall j\geq2,\ \forall(x,\xi)\in\Lambda^{j}(1),\ \lva\frac{J_{S_{j}}(x)}{J_{S^{u}(x,\xi)}(x)}-1\rva\leq C\e^{-j/C}\,,\ C=C(M)>0.\]
Here, $S^{u}$ generates the (Lagrangian) local weak instable manifold
at point $(x,\xi)$. Moreover, theses Jacobians decay exponentially
with $j$ as $j\to\infty$. This means that uniformly with respect
to $n$, \[
\prod_{j=0}^{n-1}J_{S_{n-j}}(\phi_{S_{n}}^{-j}(x))\leq C(M)\prod_{j=0}^{n-1}J_{S^{u}(\Phi^{-j}(x,\xi))}(\phi_{S_{n}}^{-j}(x))\,.\]
The Jacobian $J_{S^{u}(x,\xi)}(x)$ measures the contraction of $\Phi^{-1}$
along the unstable subspace $E^{u}(\Phi^{1}(\rho))$, where $\Phi^{1}(\rho)=(x,\xi)$,
and $x\in M$ serves as coordinates to compute this Jacobian (via
the projection $\pi$). The unstable Jacobian $J^{u}(\rho)\defeq|\det\lp d\Phi^{-1}|_{E^{u,0}(\Phi(\rho))}\rp|$
defined in Section \ref{sec:flows} express also this contraction,
but in different coordinates: for $n$ large enough, the above inequality
can then be extended to \begin{equation}
\prod_{j=0}^{n-1}J_{S_{n-j}}(\phi_{S_{n}}^{-j}(x))\leq C\prod_{j=0}^{n-1}J_{S^{u}(\Phi^{-j}(x,\xi))}(\phi_{S_{n}}^{-j}(x))\leq\ti C\prod_{j=0}^{n-1}J^{u}(\Phi^{-j}(\rho))\,.\label{eq:Prod Jac}\end{equation}
where $C,\ti C$ only depends on $M$. As noted above, because of
the Anosov property of the geodesic flow, the above products decay
exponentially with $n$. Together with the fact that the damping function
is positive, it follows that the right hand side in \eqref{eq:DecayProd1}
also decay exponentially with $n$. Recall now that $1\leq n\leq Nt_{0}$
and $N=T\log\h^{-1}$. Using \eqref{eq:DecayProd2}, we then see that
the remainders $R_{K}^{n}$ in \eqref{eq:AnsatzV} are uniformly bounded
: they satisfy \[
\|R_{K}^{n}\|\leq C_{K}\,,\quad C_{K}=C^{(K)}(M,\chi,\cV)\]
uniformly in $n$ and $z_{0}$, the point on which $\delta_{\chi,z_{0}}$
was based. From the very construction of $b^{n}(t,x)$, we then have
\begin{equation}
\|\cU^{1}\PP n\dots\cU^{1}\PP 1\cU^{1}\delta_{\chi}-\frac{1}{(2\pi\h)^{d/2}}b^{n}(1,\cdot)\|\leq C_{K}\h^{K-d/2}.\label{eq:Approx Finale}\end{equation}
But the bounds on the symbols $b_{k}^{n}$, $k>0$ given in Lemma
\ref{lem:Cnorms} tells us that \eqref{eq:Approx Finale} also holds
if we replace the full symbol $b^{n}$ by the principal symbol $b_{0}^{n}$,
provided $\h$ is chosen small enough -- say $\h\leq\h_{0}(\eps)$.
Hence, for $\h\leq\h_{0}$, \begin{eqnarray*}
\|\cU^{1}\PP n\dots\cU^{1}\PP 1\cU^{1}\delta_{\chi}\| & \leq & (2\pi\h)^{-\frac{d}{2}}\|b_{0}^{n}(1,\cdot)\|+C_{K}\h^{K-d/2}\,.\end{eqnarray*}
Now, using \eqref{eq: Ppal Symbol}, \eqref{eq:Prod Jac} and the
fact that $|\FF_{\gamma}|\leq1$, we conclude that for $a^{u}$ as
in \eqref{eq:Fonction Pression}, \[
\|b_{0}^{n}(1,x)\|\leq C\e^{n\cO(\h)}\sup_{x\in\pi\Lambda^{n}(1)}\exp\sum_{j=1}^{n}a^{u}\circ\Phi^{-j}(x,d_{x}S_{n}(1,x))\]
Here, $C=C(M)$ depends only on the manifold $M$. Let us consider
now the particular case $n=Nt_{0}$ with $N=T\log\h^{-1}$. It follows
immediately that \[
\sup_{x\in\pi\Lambda^{Nt_{0}}(1)}\exp\sum_{j=1}^{Nt_{0}}a^{u}\circ\Phi^{-j}(x,d_{x}S_{n}(1,x)\leq\prod_{k=1}^{N}\sup_{\rho\in\cW_{\beta_{k}}}\lp\exp\sum_{j=0}^{t_{0}-1}a^{u}\circ\Phi^{j}(\rho)\rp.\]
By the superposition principle already mentionned in \eqref{eq:superposition},
we then obtain for some $C=C(M)>0$ depending only on $M$:\begin{eqnarray*}
\|\cU\PP{Nt_{0}}\dots\cU\PP 1\cU^{1}\oph(\chi)\| & \leq & C\sum_{\ell}\sup_{z_{0}}\|\cU\PP{Nt_{0}}\dots\PP 1\cU^{1}\delta_{z,\alpha_{0}}^{\ell}\|\\
 & \leq & C\h^{-d/2}\|b_{0}^{n}\|+C_{K}\h^{K-d/2}\\
 & \leq & C\h^{-\frac{d}{2}}\prod_{k=1}^{N}\sup_{\rho\in\cW_{\beta_{k}}}\lp\exp\sum_{j=0}^{t_{0}-1}a^{u}\circ\Phi^{j}(\rho)\rp\end{eqnarray*}
To get the last line, we have noticed that $K$ can be chosen arbitrary
large: since $n\leq Tt_{0}\log\h^{-1}$, we see that for $\h$ small
enough, the main term in the right hand side of the second line is
larger than the remainder $C_{K}\h^{K-d/2}$, and $\e^{Nt_{0}\cO(\h)}=\cO(1)$.
This completes the proof of Proposition \ref{pro:HypDispEst}. 

\appendix

\section{Semiclassical analysis on compact manifolds\label{sec:semiclass}}

In this appendix we gather standard notions of pseudodifferential
calculus on a compact, $d$ dimensional manifold $M$ endowed with
a Riemannian structure coming from a metric $g$. As usual, $M$ is
equiped with an atlas $\{f_{\ell},V_{\ell}\}$, where $\{V_{\ell}\}$
is an open cover of $M$ and each $f_{\ell}$ is a diffeomorphism
form $V_{\ell}$ to a bounded open set $W_{\ell}\subset\R^{d}$. Functions
on $\R^{d}$ can be pulled back via $f_{\ell}^{*}:C^{\infty}(W_{\ell})\to C^{\infty}(V_{\ell})$.
The canonical lift of $f_{\ell}$ between $T^{*}V_{\ell}$ and $T^{*}W_{\ell}$
is denoted by $\ti f_{\ell}$:\[
(x,\xi)\in T^{*}V_{\ell}\mapsto\ti f_{\ell}(x,\xi)=(f_{\ell}(x),(Df_{\ell}(x)^{-1})^{T}\xi)\in T^{*}W_{\ell}\,,\]
where $A^{T}$ denotes the transpose of $A$. Its corresponding pull-back
will be denoted by $\ti f_{\ell}^{*}:C^{\infty}(T^{*}W_{\ell})\to C^{\infty}(T^{*}V_{\ell})$.
A smooth partition of unity adapted to the cover $\{V_{\ell}\}$ is
a set of functions $\phi_{\ell}\in C_{c}^{\infty}(V_{\ell})$ such
that $\sum_{\ell}\phi_{\ell}=1$ on $M$. 

Any observable (i.e. a function $a\in C^{\infty}(T^{*}M)$) can now
be split into $a=\sum_{\ell}a_{\ell}$ where $a_{\ell}=\phi_{\ell}a$,
and each term pushed to $\ti a_{\ell}=(\ti f_{\ell}^{-1})^{*}a_{\ell}\in C^{\infty}(T^{*}W_{\ell})$.
If $a$ belongs to a standard class of symbols, for instance \[
a\in S^{m,k}=S^{k}(\langle\xi\rangle^{m})\defeq\left\{ a=a_{\h}\in C^{\infty}(M),\ |\d_{x}^{\alpha}\d_{\xi}^{\beta}a|\leq C_{\alpha,\beta}\h^{-k}\langle\xi\rangle^{m-|\beta|}\right\} ,\]
 each $a_{\ell}$ can be be Weyl-quantized into a pseudodifferential
operator on $\cS(\R)$ via the formula \[
\forall u\in\cS(\R^{d}),\ \Op_{\h}^{w}(\ti a_{\ell})u(x)=\frac{1}{(2\pi\h)^{d}}\int\e^{\frac{\i}{\h}\langle x-y,\xi\rangle}\ti a_{\ell}\lp\frac{x+y}{2},\xi;\h\rp u(y)\: dy\: d\xi\]
To pull-back this operator on $C^{\infty}(V_{\ell})$, one first takes
another smooth cutoff $\psi_{\ell}\in C_{c}^{\infty}(V_{\ell})$ such
that $\psi=1$ in a neighbourhood of $\supp\phi_{\ell}$. The quantization
of $a\in S^{m,k}$ is finally defined by gluing local quantizations
together, yielding to \[
\forall u\in C^{\infty}(M),\oph(a)u=\sum_{\ell}\psi_{\ell}\times f_{\ell}^{*}\circ\oph^{w}(\ti a_{\ell})\circ(f_{\ell}^{-1})^{*}(\psi_{\ell}u)\]
The space of pseudodifferential operators obtained from $S^{k,m}$
by this quantization will be denoted by $\Psi^{m,k}$. Although this
quantization depends on the cutoffs, the principal symbol map $\sigma:\Psi^{m,k}\to S^{m,k}/S^{m,,k-1}$
is intrinsically defined and do not depend on the choice of coordinates.
The residual class is made of operators in the space $\Psi^{m,-\infty}$.
As an example, the (semiclassical) Laplacian $-\h^{2}\Dg\in\Psi^{0,2}$
is a pseudodifferential operator, and its principal symbol is given
by $\sigma(-\h^{2}\Dg)=\|\xi\|_{g}^{2}=g_{x}(\xi,\xi)\in S^{2,0}$. 

In this article, we are concerned with a purely semiclassical theory
and then deal only with compact subsets of $T^{*}M$. If $A\in\Psi^{m,k}$,
we will denote by $\WF_{\h}(A)$ the semiclassical wave front set
of $A$. A point $\rho\in T^{*}M$ belongs to $\WF{}_{\h}(A)$ if
for some choice of local coordinates near the projection of $\rho,$
the full symbol of $A$ is in the class $S^{m,-\infty}$. $\WF_{\h}(A)$
is a closed subset of $T^{*}M$, and $\WF_{\h}(AB)\subset\WF_{\h}(A)\cap\WF_{\h}(B)$.
In particular, if $\WF_{\h}(A)=\emptyset$, then $A$ is a negligible
operator, i.e. $A\in\Psi^{m,-\infty}$. If $\Psi\in L^{2}(M)$, we
also define the semiclassical wave front set of $\Psi$ by :\[
\WF_{\h}(\Psi)=\left\{ (x,\xi)\ :\ \exists a\in S^{m,0},\ a(x,\xi)\neq0,\ \|\Op_{\h}(a)\Psi\|_{L^{2}(M)}=\cO(\h^{\infty})\right\} ^{\bold c}\]
where the superscript $^{\bold c}$ indicates the complementary set.
We will often make use of the following fundamental propagation property
: if $U^{t}$ is a Fourier integral operator associated to a symplectic
diffeomorphism $\Phi^{t}:T^{*}M\to T^{*}M$, then \[
\WF_{\h}(U^{t}\Psi)=\Phi^{t}(\WF_{\h}(\Psi)).\]

\section*{Acknowledgements}

I am very grateful to Stéphane Nonnenmacher for numerous helpful discussions
on the subject of the present paper. I would also like to thank sincerely
Maciej Zworski for interesting suggestions, and the UC Berkeley for
its hospitality during March and April 2008.
\end{document}